\documentclass[conference]{IEEEtran}
\usepackage{url}
\usepackage[cmex10]{amsmath}
\usepackage{amsthm}

\newtheorem{theorem}{Theorem}
\newtheorem{property}{Property}
\usepackage{changepage}
\usepackage[noend]{algorithmic}
\usepackage{algorithm2e}
\usepackage{multirow}
\usepackage{xcolor}
\usepackage{color}
\usepackage{times}
\usepackage{mathptmx}
\usepackage{mathtools}
\usepackage{mathrsfs}
\usepackage{enumerate}
\usepackage{amssymb}
\usepackage{subfigure, float}

\DeclareMathAlphabet{\mathcal}{OMS}{cmsy}{m}{n}

\newcommand{\nop}[1]{}

\newcommand{\todo}[1]{\textcolor{red}{[\sc #1]}}
\newcommand{\updates}[1]{\textcolor{black}{#1}}
\newcommand{\tabincell}[2]{\begin{tabular}{@{}#1@{}}#2\end{tabular}}

\newcommand{\NA}{---}

\newtheorem{corollary}{Corollary}

\usepackage{cite}

\ifCLASSINFOpdf
\else
\fi
\begin{document}
%
\title{Mining Density Contrast Subgraphs}

\author{\IEEEauthorblockN{Yu Yang$^{1}$, Lingyang Chu$^{1}$, Yanyan Zhang$^{1}$, Zhefeng Wang$^{2}$, Jian Pei$^{1}$ and Enhong Chen$^{2}$}
\IEEEauthorblockA{
$^{1}$Simon Fraser University, Burnaby, Canada\\
$^{2}$University of Science and Technology of China, Hefei, China\\
\{yya119,lca117,yanyanz\}@sfu.ca, zhefwang@mail.ustc.edu.cn, jpei@cs.sfu.ca, cheneh@ustc.edu.cn}
}

\maketitle

\begin{abstract}
Dense subgraph discovery is a key primitive in many graph mining applications, such as detecting communities in social networks and mining gene correlation from biological data. Most studies on dense subgraph mining only deal with one graph. However, in many applications, we have more than one graph describing relations among a same group of entities. In this paper, given two graphs sharing the same set of vertices, we investigate the problem of detecting subgraphs that contrast the most with respect to density. We call such subgraphs Density Contrast Subgraphs, or DCS in short. Two widely used graph density measures, average degree and graph affinity, are considered. For both density measures, mining DCS is equivalent to mining the densest subgraph from a ``difference'' graph, which may have both positive and negative edge weights. Due to the existence of negative edge weights, existing dense subgraph detection algorithms cannot identify the subgraph we need. We prove the computational hardness of mining DCS under the two graph density measures and develop efficient algorithms to find DCS. We also conduct extensive experiments on several real-world datasets to evaluate our algorithms. The experimental results show that our algorithms are both effective and efficient.
\end{abstract}


%
\IEEEpeerreviewmaketitle

\section{Introduction}\label{sec:intro}
Dense subgraph extraction lies at the core of graph data mining. The problem and its variants have been intensively studied. Most of the existing studies focus on finding the densest subgraph in one network. For example, polynomial time algorithms and efficient approximation algorithms are devised to find the subgraph with maximum average degree~\cite{goldberg1984finding}. There are also quadratic programming methods for extracting subgraphs with high graph affinity density~\cite{liu2013fast,pavan2007dominant}.

In many real-world applications, there are often more than one kind of relations among objects studied. Thus, it is common to have more than one graph describing a same set of objects, one kind of relation captured by one graph. As a result, an interesting contrast data mining problem arises. Given two graphs sharing the same set of vertices, what is the subgraph such that the gap between its density in the two graphs is the largest? We call such a subgraph the \textbf{Density Contrast Subgraph} (\textbf{DCS}). 

\updates{To demonstrate the power of DCS, consider the task of surveying and summarizing the trends of an area, say data mining research. Such a task is practical and common for technical writers, academic researchers and graduate students among many others. Based on a database of published data mining papers, how can we detect trends from the database automatically? Angel~\textit{et~al.}~\cite{angel2012dense} proposed to build a keyword association graph from the input text data, and identify stories/topics via groups of densely-connected keywords from it. For example, applying the method of~\cite{angel2012dense} on data mining papers we may find a topic ``scalable tensor factorization'', because the words ``scalable'', ``tensor'' and ``factorization'' often co-occur in papers. However, directly extracting dense subgraphs corresponding to densely-connected keywords from a keyword association graph like~\cite{angel2012dense} may not help us detect trends effectively. For example, in our experiments, if we just extract dense subgraphs from the graph indicating pairwise keywords association strength in the titles of data mining papers published in the last 10 years, we find topics ``time series'' and ``feature selection''. But these two topics have been intensively investigated ever since and do not present a new trend. In the recent 10 years, according to the graph density measure on the data we have, the topic ``time series'' even cooled down a little bit.}

\updates{To detect trends effectively, we take two keyword association graphs into consideration and apply DCS algorithms. Besides the keyword association graph based on papers published recently, we also need the other keyword association graph derived from the papers published in early years. Those groups of keywords whose connection strengths are much tighter in the recent keyword association graph than in the early keyword association graph are identified as trends in data mining. In our experiments we obtained results like ``social networks'', ``matrix factorization'' and ``unsupervised feature selection''. These topics all became popular only in recent years. }

\updates{DCS can also be applied to detecting current anomalies against historical data. Specifically, we can build a weighted graph where the edge weights are our expectation of how tightly the vertices are connected to each other, which can be derived from, for example, historical data. Then, we observe the current pairwise connection strength of vertices, and build another weighted graph based on our observations. We apply DCS on these two weighted graphs. Some typical application scenarios include detecting emerging traffic hotspot clutters, emerging communities in social networks, and money launderer dark networks.}

In this paper, we study the Density Contrast Subgraph problem under two widely adopted graph density measures, average degree and graph affinity. One may notice that for both density measures, we may form a ``difference'' graph, where the weight of each edge is obtained by the difference of the weights of this edge in the two graphs. However, this does not mean that traditional densest subgraph extraction methods can be applied to find density contrast subgraphs. In the traditional densest subgraph problems, edge weights are always positive. In the difference graph of the density contrast subgraph problem, we may have negative edge weights. The existence of negative edge weights changes the nature of densest subgraph finding substantially. For example, finding the densest subgraph with respect to average degree in a graph with only positive edge weights is polynomial time solvable~\cite{goldberg1984finding}, and has an efficient 2-approximation algorithm~\cite{charikar2000greedy}, while if the graph has negative edge weights, it becomes NP-hard and also hard to approximate as to be proved in Section~\ref{sec:ad}.

To tackle the Density Contrast Subgraph problem, we make several technical contributions. We prove the computational hardness of finding DCS under the two density measures of average degree and graph affinity. For the average degree measure, we also prove it is hard to approximate within a factor of $O(n^{1-\epsilon})$. An efficient $O(n)$-approximation algorithm is then developed to solve this problem. The DCS problem under the graph affinity measure is also NP-hard, and is a QP (Quadratic Programming) which is non-concave. For this problem, we first devise an efficient 2-Coordinate Descent algorithm that is guaranteed to converge to a KKT point. Based on the 2-coordinate descent algorithm, we give a constructive proof of that edges of negative weights cannot appear in a DCS with respect to graph affinity. Using our construction, we can further improve a KKT point solution to a positive clique solution. A smart initialization heuristic is proposed to reduce the number of initializations for our iterative algorithm, which in experiments brings us  speedups of 1-3 orders of magnitude. Extensive empirical studies are conducted to demonstrate the effectiveness and efficiency of our algorithms. 

The rest of the paper is organized as follows. We review the related work in Section~\ref{sec:related}. In Section~\ref{sec:pre}, we briefly introduce the two density measures used in our work, average degree and graph affinity, and formulate the Density Contarst Subgraph problem. In Section~\ref{sec:ad}, we give our solutions to the DCS problem under the measure of average degree. In Section~\ref{sec:ga}, we tackle the DCS problem under the graph affinity measure. We report the experimental results in Section~\ref{sec:exp} and conclude the paper in Section~\ref{sec:con}.

\section{Related Work}\label{sec:related}
Dense subgraph extraction is a key problem in both algorithmic graph theory and graph mining applications~\cite{tsourakakis2013denser, angel2012dense, mitzenmacher2015scalable, khuller2009finding, fratkin2006motifcut}. One of the most popular definitions of subgraph density is the average degree. Intensive studies have been conducted on finding a subgraph with the maximum average degree in one single graph~\cite{bhattacharya2015space, epasto2015efficient, goldberg1984finding, charikar2000greedy}. Goldberg~\cite{goldberg1984finding} first proposed a polynomial time algorithm based on maximum flow.  Charikar~\cite{charikar2000greedy} described a simple greedy algorithm which has an approximation ratio of 2. 

Besides average degree, graph affinity, which is a quadratic function $\textbf{x}^\top A \textbf{x}$ of a subgraph embedding $\textbf{x} \in \triangle^n$, is also widely adopted as a measure of subgrah density~\cite{liu2013fast,pavan2007dominant,chu2015alid,wang2016tradeoffs}. Motzkin and Straus~\cite{motzkin1965maxima} proved that, for unweighted graphs, maximizing graph affinity is equivalent to finding the maximum clique in the graph. Pavan and Pelillo~\cite{pavan2007dominant} first proposed an algorithm based on replicator dynamics to find local maximas of the quadratic function $\textbf{x}^\top A \textbf{x}$ on the simplex $\triangle^n$. Liu~\textit{et~al.}~\cite{liu2013fast} proposed a highly efficient algorithm called SEA (see Appendix) to solve the problem, where the core idea is to use a shrink-and-expand strategy to accelerate the process of finding Karush-Kuhn-Tucker (KKT) points. Wang~\textit{et~al.}~\cite{wang2016tradeoffs} discussed the trade-off between the graph affinity density and subgraph size in extracting dense subgraphs.

Please note that the existing work on maintaining dense subgraphs on temporal graphs~\cite{angel2012dense,bhattacharya2015space, epasto2015efficient, bahmani2012densest} cannot solve our problem, although two consecutive snapshots of a temporal graph can be regarded as a special case of the input to DCS. \cite{angel2012dense, bhattacharya2015space, epasto2015efficient, bahmani2012densest} are all for extracting dense subgraphs from the latest snapshot, where for a valid input there are no edges with negative weights. The algorithms in~\cite{bhattacharya2015space, epasto2015efficient} can only deal with unweighted graphs. In our problem, mining DCS from two graphs is equivalent to mining dense subgraphs from a ``difference graph'', which may have negative edge weights. We show in Section~\ref{sec:ad} that, when the density measure is average degree, the existence of negative edge weights makes our DCS problem NP-hard and hard to approximate. This is dramatically different from extracting densest subgraph with respect to average degree~\cite{bhattacharya2015space, epasto2015efficient, bahmani2012densest}, which is polynomial time solvable and has an efficient 2-approximation algorithm. For the graph affinity density measure, ~\cite{angel2012dense} considers a general definition of subgraph density where edge density, the discrete version of graph affinity, is used as a special case. However, ~\cite{angel2012dense} is for maintaining all subgraphs whose density is greater than a threshold, and only subgraphs with size (\#vertices) smaller than $N_{max}=5$ is considered. Thus, although mining DCS with respect to graph affinity can be reduced to finding the densest subgraph in one single graph (the ``difference'' graph), techniques in~\cite{angel2012dense} still cannot be used. 

Mining dense subgraphs from multiple networks to find ``coherent'' dense subgraphs also attracts much research interest~\cite{hu2005mining,li2012pattern,kelley2005systematic,wu2016mining}. For example, Wu~\textit{et~al.}~\cite{wu2016mining} investigated the problem of finding a subgraph that is dense in one conceptual network and also connected in a physical network. All these studies focus on finding ``coherent'' dense subgraphs in multiple graphs.
Another line of related research is contrast graph mining, which aims at discovering subgraphs that manifest drastic differences between graphs. \updates{Wang~\textit{et~al.}~\cite{wang2008spatial} and Gionis~\textit{et~al.}~\cite{gionis2015bump} studied how to find the anomalous subgraphs that contrast others in one graph.} Ting and Bailey~\cite{ting2006mining} proposed algorithms to find the minimal contrast subgraph, which is a graph pattern appears in one graph but not in the other graph, and all of its proper subgraphs are either shared by or not contained in the two graphs. Yang~\textit{et~al.}~\cite{yang2014mining} studied the problem of detecting the most frequently changing subgraph in two consecutive snapshots of a time-evolving graph. The major difference between these studies and our work is that, none of these studies adopt subgraph density as the measure for mining contrast subgraphs. 

\nop{
Cadena~\textit{et~al.}~\cite{cadena2016dense} investigated how to extract the subgraph whose total edge weight deviates from its expected total edge weight the most, which is the work closest to ours in literature. Although the total edge weight of a subgraph is related to our density measures, average degree and graph affinity\footnote{The total edge weight of a subgraph is the numerator of this subgraph's average degree and edge density, which is often regarded as the discrete version of graph affinity~\cite{liu2013fast,wang2016tradeoffs}.}, these three measures are still quite different from each other. Thus, properties of and solutions to the problem in~\cite{cadena2016dense} and our problems are very different. 
}
\updates{~\cite{cadena2016dense} is the work closest to ours in literature. In~\cite{cadena2016dense}, Cadena~\textit{et~al.} investigated how to extract the subgraph whose total edge weight deviates from its expected total edge weight the most. The total edge weight is related to the density measures adopted in our work, average degree and graph affinity, since the total edge weight of a subgraph is the numerator of this subgraph's average degree and edge density, which is often regarded as the discrete version of graph affinity~\cite{liu2013fast,pavan2007dominant,wang2016tradeoffs}. However, these three measures are still quite different from each other. Thus, properties of the problem in~\cite{cadena2016dense} and our problems are very different. }

\section{Preliminaries}\label{sec:pre}
In this section, we introduce several essential concepts in our discussion and formulate the Density Contrast Subgraph problem. For readers' convenience, Table~\ref{tab:notation} lists the frequently used notations.

\begin{table}[t]
\centering
\begin{tabular}{|p{30mm}|p{50mm}|}
\hline
Notation & Description \\ \hline
$G=\langle V,E,A\rangle$ & An undirected and weighted graph, where each edge $(u,v) \in E$ is associated with a positive weight $A(u,v)$\\ \hline
$G(S)=\langle V,E(S),A(S) \rangle$    & The induced subgraph of $S$ in graph $G$  \\ \hline
$W(S)$    &  The total degree of $S$ in graph $G$. $W(S)=\sum_{(u,v) \in E(S)}{A(u,v)}$ \\ \hline
$\triangle^n$ & A simplex. $\triangle^n=\{\textbf{x} \mid \sum_{i=1}^n{x_i}=1, x_i \geq 0\}$ \\ \hline
$\textbf{x} \in \triangle^n$ & An embedding of a subgraph, $x_u$ denotes the participation of $u$ in this subgraph\\ \hline
$S_{\textbf{x}}$ & Support set of $\textbf{x}$. $S_{\textbf{x}}=\{u \mid x_u > 0\}$\\ \hline
$G_1=\langle V,E_1,A_1 \rangle$, $G_2=\langle V,E_2,A_2 \rangle$ & Inputs of our Density Contrast Subgraph problem \\ \hline
$G_D=\langle V,E_D,D \rangle$ & The difference graph between $G_2$ and $G_1$, where $D=A_2-A_1$ and $E_D=\{(u,v) \mid D(u,v) \neq 0 \}$\\ \hline
$G_{D^+}=\langle V,E_{D^+},D^+ \rangle$ & The ``positive'' part of $G_D$, where $D^+(i,j)=\max\{D(i,j),0\}$ and $E_{D^+}=\{(u,v) \mid D(u,v) > 0 \}$\\ \hline
$N_D(i)$ & The set of $i$'s neighbors in $G_D$. $N_D(i)=\{j \mid D(i,j) \neq 0\}$ \\ \hline
$W_D(i;G_D(S))$ & The degree of vertex $i$ in the induced subgraph $G_D(S)$. $W_D(i;G_D(S))=\sum_{j \in N_D(i) \cap S}{A(i,j)}$ \\ \hline
\end{tabular}
\caption{Frequently used notations.}
\label{tab:notation}
\end{table}

\subsection{Measures of Graph Density}
An undirected and weighted graph is represented by $G=\langle V,E,A \rangle$, where $V$ is a set of vertices, $E$ is a set of edges and $A$ is an affinity matrix. Since $G$ is undirected, if $(u,v) \in E$ then $(v,u) \in E$. Denote by $n=|V|$ the number of vertices and $m=|E|$ the number of edges. $A$ is an $n \times n$ symmetric matrix. The entry $A(u,v)>0$ denotes the weight of the edge $(u,v)$, and $A(u,v)=0$ if $(u,v) \notin E$. Given an undirected graph $G=\langle V,E,A \rangle$ and a subset of vertices $S$, the induced subgraph of $S$ is denoted by $G(S)=\langle S,E(S),A(S) \rangle$, where $E(S)=\{(u,v) \mid (u,v) \in E \wedge u \in S \wedge v \in S\}$, and $A(S)$ is a submatrix of $A$ so that only the row and columns of vertices in $S$ are present. 

Average degree is a widely investigated graph density measure. Given an undirected graph $G=\langle V,E,A \rangle$ and a set of vertices $S$, the total degree of the induced subgraph $G(S)$ is $W(S)=\sum_{(u,v) \in E(S)}{A(u,v)}$. The \textbf{average degree} of the induced subgraph $G(S)$ is defined by
\begin{equation}\label{eq:ad}
	\rho(S)=\frac{W(S)}{|S|}=\frac{1}{|S|}\sum_{u \in S}{\sum_{(u,v) \in E(S)}{A(u,v)}}=\frac{1}{|S|}\sum_{u \in S}{W(u;G(S))}
\end{equation}
where $W(u;G(S))=\sum_{(u,v) \in E(S)}{w(u,v)}$ is $u$'s degree in $G(S)$. 

Graph affinity is another popularly adopted graph density measure. In graph affinity, a subgraph is represented by a \textbf{subgraph embedding} in a standard simplex $\triangle^n=\{\textbf{x} \mid \sum_{i=1}^n{x_i}=1, x_i \geq 0\}$. For a subgraph embedding $\textbf{x}=[x_1,x_2,...,x_n]$, the entry $x_u$ indicates the participation importance of vertex $u$ in the subgraph. Denote by $S_{\textbf{x}}=\{u \mid x_u > 0\}$ the \textbf{support set} of $\textbf{x}$. The \textbf{graph affinity} density of a subgraph embedding $\textbf{x} \in \triangle^n$ is defined by
\begin{equation}\label{eq:ga}
	f(\textbf{x})=\textbf{x}^{\top}A\textbf{x}=\sum_{i=1}^{n}\sum_{j=1}^n{x_ix_jA(i,j)}=\sum_{(u,v) \in E_S}{x_ux_vA(u,v)}
\end{equation}

In traditional dense subgraph mining problems, when the density measure is average degree, the densest subgraph is often large in size~\cite{angel2012dense}, while if the density measure is graph affinity, the support set of the densest subgraph embedding is normally small~\cite{wang2016tradeoffs}.

\subsection{Mining Density Contrast Subgraph}
Given two undirected graphs $G_1=\langle V,E_1,A_1 \rangle$ and $G_2=\langle V,E_2,A_2 \rangle$, we want to find a subgraph such that its density in $G_1$ minus its density in $G_1$ is a large value. Similar to traditional dense subgraph mining, we are more interested in the subgraph whose density difference is the greatest among all subgraphs. Thus, the \textbf{Density Contrast Subgraph} (\textbf{DCS}) problem can be formulated as an optimization problem. Specifically, if the density measure is the average degree, the optimization problem is

\begin{equation}\label{eq:max_ad}
	\max_{S \subseteq V}{\rho_2(S)-\rho_1(S)=\frac{W_2(S)}{|S|}-\frac{W_1(S)}{|S|}}
\end{equation}
where $W_1(S)$ and $W_2(S)$ are total degrees of $S$ in $G_1$ and $G_2$, respectively. We call Eq.~\ref{eq:max_ad} the problem of \textbf{Density Contrast Subgraphs with respect to Average Degree} (\textbf{DCSAD}).

If we adopt graph affinity as the density measure, the optimization problem then becomes
\begin{equation}\label{eq:max_ga}
	\max_{\textbf{x} \in \triangle^n}{f_2(\textbf{x})-f_1(\textbf{x})=\textbf{x}^{\top}A_2\textbf{x}-\textbf{x}^{\top}A_1\textbf{x}}
\end{equation}
We call Eq.~\ref{eq:max_ga} the problem of \textbf{Density Contrast Subgraphs with respect to Graph Affinity} (\textbf{DCSGA}). \updates{Note that, similar to maximizing graph affinity in $G$, which is often used to maximize the edge density ($\frac{W(S)}{|S|^2}$)~\cite{wang2016tradeoffs}, we can also solve (\textbf{DCSGA}) for finding a subgraph whose edge density gap in $G_1$ and $G_2$ ($\frac{W_D(S)}{|S|^2}$) is maximized. To convert the solution $\textbf{x}$ to a set of vertices $S$, we just set $S=S_{\textbf{x}}$.}

It is easy to find that to find the subgraph such that the absolute value of its density difference is maximized, besides solving Eq.~\ref{eq:max_ad} or Eq.~\ref{eq:max_ga}, we also solve $\max_{S \subseteq V}{\rho_1(S)-\rho_2(S)}$ or $\max_{\textbf{x} \in \triangle^n}{f_1(\textbf{x})-f_2(\textbf{x})}$.

\begin{figure}
    \centering
    \includegraphics[width=.2\textwidth]{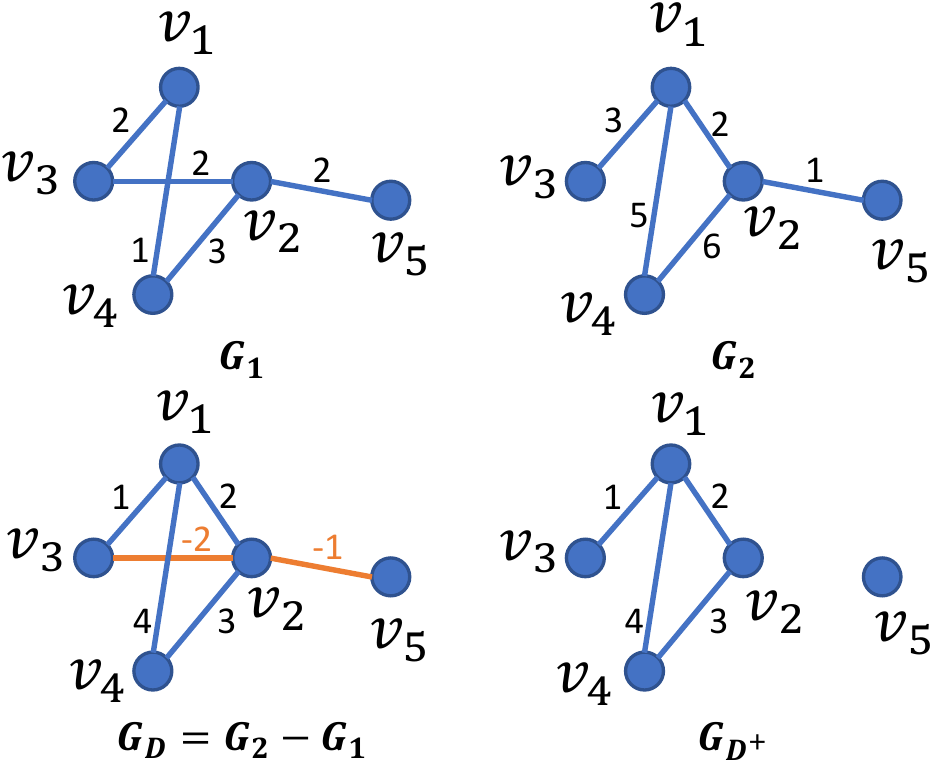}
    \caption{An Example of the Difference Graph}
    \label{fig:exp}
\end{figure}

A nice property that both Eq.~\ref{eq:max_ad} and Eq.~\ref{eq:max_ga} have is that the objective equals the density of $S$'s induced subgraph (or $\textbf{x}$ if using graph affinity as density) in a ``difference graph'' between $G_2$ and $G_1$. Given $G_1=\langle V,E_1,A_1 \rangle$ and $G_2=\langle V,E_2,A_2 \rangle$, the \textbf{difference graph} $G_D=\langle V,E_D,D \rangle$ is the graph associated with the affinity matrix $D=A_2-A_1$. Thus, $E_D=\{(u,v) \mid D(u,v) \neq 0 \}$. We also define the graph that contains only edges with positive  weights as $G_{D^+}=\langle V, E_{D^+}, D^+ \rangle$, where $E_{D^+}=\{(u,v) \mid D(u,v) > 0 \}$. Fig.~\ref{fig:exp} gives an example of $G_1$, $G_2$, $G_D$ and $G_{D^+}$.
It is easy to verify that Eq.~\ref{eq:max_ad} is equivalent to
\begin{equation}\label{eq:DCS_ad}
	\max_{S \subseteq V}{\rho_D(S)=\frac{W_D(S)}{|S|}}
\end{equation} 
where $W_D(S)$ is the total degree of $S$ in $G_D$. Also, Eq.~\ref{eq:max_ga} is equivalent to
\begin{equation}\label{eq:DCS_ga}
	\max_{\textbf{x} \in \triangle^n}{f_D(\textbf{x})=\textbf{x}^{\top}D\textbf{x}}
\end{equation}

The major difference between finding dense subgraphs in a difference graph $G_D$ and the traditional dense subgraph detection problems is that there are negative edge weights in a difference graph. In Sections~\ref{sec:ad} and~\ref{sec:ga} we will analyze how negative edge weights affect properties and algorithms of dense subgraph mining problems.

Also, from Eq.~\ref{eq:DCS_ad} and Eq.~\ref{eq:DCS_ga} we can see that the optimal value is positive if and only if the matrix $D$ has at least one positive entry, that is, the difference graph has at least one edge with potisive weight. If $D$ does not have positive entries, the optimal values to Eq.~\ref{eq:DCS_ad} and Eq.~\ref{eq:DCS_ga} are both 0, the optimal $S$ to Eq.~\ref{eq:DCS_ad} contains one single vertex, and the optimal $\textbf{x}$ to Eq.~\ref{eq:DCS_ga} has only one entry that equals 1 and all other entries 0.

\subsection{Why not Ratio of Difference?}

Instead of the absolute value of density difference, why don't we consider the ratio of density difference, i.e. $\frac{\rho_2(S)}{\rho_1(S)}$ or $\frac{f_2(\textbf{x})}{f_1(\textbf{x})}$, as the objective? The reason is that the ratio of density difference sometimes is not well-defined or has trivial solutions. Consider a single vertex $u$ as a subgraph. Its densities in $G_1$ and $G_2$ are both 0 so the ratio of density difference is $\frac{0}{0}$. Also, in Fig.~\ref{fig:exp}, the edge $(v_1,v_2)$ has density ratio $+\infty$ since it only appears in $G_2$ but not $G_1$.

\subsection{Generalization of the Difference Graph}
In Sections~\ref{sec:ad} and~\ref{sec:ga} we will introduce our DCS finding algorithms that can take any weighted graphs as input, where the weight of an edge can be positive or negative. Thus, the definition of the difference graph of $G_1=\langle V,E_1,A_1 \rangle$ and $G_2=\langle V,E_2,A_2 \rangle$ is not restricted to the graph whose affinity matrix is $A_2-A_1$. For example, we can set the difference graph as $G_D=\langle V,E_D,D=A_2-\alpha A_1\rangle$ and maximizing $\rho_D(S)$ (or $f_D(\textbf{x})$) is equivalent to finding $S$ (or $\textbf{x}$) such that $\rho_2(S) \geq \alpha \rho_1(S)$ (or $f_2(\textbf{x}) \geq \alpha f_1(\textbf{x})$), and $\rho_2(S)-\alpha \rho_2(S)$ (or $f_2(\textbf{x})-\alpha f_1(\textbf{x})$) is maximized. This is similar to the optimal $\alpha$-quasi-clique problem~\cite{tsourakakis2013denser}. Also, when there is one edge in the difference graph whose weight is much heavier than all the other edges, such an edge itself is very possible to be the optimal subgraph. To avoid this, for edges with too heavy weights in $G_D$, we can adjust their weights such that they are not too much heavier than other edge weights in $G_D$. Then the DCS extracted usually will become larger in size.  

\section{DCS with respect to Average Degree}\label{sec:ad}
In this section, we first explore some key properties of the \textbf{DCSAD} problem.
Then, we devise an efficient greedy algorithm with \updates{a data-dependent} ratio.

\subsection{Complexity and Approximability}
Like traditional dense subgraph discovery problem, the \textbf{DCSAD} prefers ``connected'' subgraphs, of course, in the difference graph $G_D$. 

\begin{property}\label{prp:connect}
Let $G_D$ be the difference graph of $G_1$ and $G_2$. For any $S \subseteq V$, if $G_D(S)$ is not a connected subgraph, then there exists a set $S' \subseteq S$ such that $G_D(S')$ is connected and the density difference $\rho_D(S') \geq \rho_D(S)$.
\end{property}
\begin{proof}
	 Without loss of generality, we assume $G_D(S)$ has two connected components $G_D(S_1)$ and $G_D(S_2)$, where $S_1 \cup S_2=S$ and $S_1 \cap S_2=\emptyset$. Clearly, $W_D(S)=W_D(S_1)+W_D(S_2)$ because $G_D(S_1)$ and $G_D(S_2)$ are isolated. So we have $\rho_D(S)=\frac{W_D(S)}{|S|}=\frac{|S_1|}{|S|}\rho_D(S_1)+\frac{|S_2|}{|S|}\rho_D(S_2)$, which means $\rho_D(S)$ is a convex combination of $\rho_D(S_1)$ and $\rho_D(S_2)$. Thus, $\rho(S)=\frac{W_D(S)}{|S|} \leq \max\{\rho_D(S_1),\rho_D(S_2)\}$
\end{proof}

Traditional dense subgraph discovery with respect to average degree can be solved in polynomial time~\cite{goldberg1984finding}, and has an efficient 2-approximation algorithm~\cite{charikar2000greedy}. Unfortunately, our problem does not have the same computational properties.

\begin{theorem}\label{th:ad_hard}
The \textbf{DCSAD} (Eq.~\ref{eq:DCS_ad}) problem is NP-hard.
\end{theorem}
\begin{proof}
We prove this by a reduction from the maximum clique problem, which is known NP-hard. Given an instance of the maximum clique problem, which is an undirected and unweighted graph $G=\langle V,E \rangle$, we build two graphs $G_1$ and $G_2$ as the input of the \textbf{DCSAD} problem. Let $E_1=\{(u,v) \mid (u,v) \in V \times V \wedge u \neq v \wedge (u,v) \notin E\}$. We set $G_1=\langle V,E_1,A_1 \rangle$ and for every edge $(u,v) \in E_1$, we set the weight $A_1(u,v)=|E|+1$. Clearly building $G_1$ and $G_2$ can be done in polynomial time w.r.t.\ the size of $G$. We set $G_2=\langle V,E_2,A_2 \rangle$ where $E_2=E$. For every edge $(u,v) \in E_2$, we set the weight $A_2(u,v)=1$. 

It is obvious that for any $S \subseteq V$, the density difference ${\frac{W_2(S)}{|S|}-\frac{W_1(S)}{|S|}} < 0$ if $G_1(S)$, the induced subgraph of $S$ in $G_1$, contains at least one edge in $E_1$. Thus, the optimal $S$ must satisfy that $G_1(S)$ does not contain any edges in $E_1$. Due to the definition of $E_1$, $G_2(S)$, the induced subgraph of $S$ in $G_2$ is a clique. So the optimal density difference is $|S|-1$ where $S$ is the maximum clique in $G_2$. Because $G_2$ and $G$ actually are the same, the optimal density difference of $G_2$ and $G_1$ is at least $k-1$ if and only if $G$ contains a clique with at least $k$ vertices. Due to the NP-hardness of the maximum clique problem, the \textbf{DCSAD} problem is also NP-hard.
\end{proof}

The \textbf{DCSAD} problem is not only NP-hard but also hard to approximate under reasonable complexity assumptions.

\begin{corollary}\label{th:ad_approx}
Assuming \textbf{P}$\neq$\textbf{NP}, the \textbf{DCSAD} problem (Eq.~\ref{eq:DCS_ad}) cannot be approximated within $O(n^{1-\epsilon})$ for any $\epsilon > 0$.
\end{corollary}
\begin{proof}
We still use our reduction in the proof of Theorem~\ref{th:ad_hard}. We already proved that the optimal density difference is $k-1$ where $k$ is the size of the maximum clique in $G$. Also, it is easy to see that if a \textbf{DCSAD} algorithm returns a value $k'-1$ such that $\frac{k-1}{k'-1} \leq \beta$, there is a $k'$-clique in $G$. Since $k \geq k'$, $\frac{k}{k'} \leq \frac{k-1}{k'-1} \leq \beta$. Thus, if \textbf{DCSAD} can be approximated within $\beta$, so is the maximum clique problem.

It is known that the maximum clique problem cannot be approximated within $O(n^{1-\epsilon})$ for any $\epsilon > 0$, assuming \textbf{P}$\neq$\textbf{NP}. Thus, if \textbf{P}$\neq$\textbf{NP}, the \textbf{DCSAD} problem (Eq.~\ref{eq:max_ad}) cannot be approximated within $O(n^{1-\epsilon})$ for any $\epsilon > 0$.
\end{proof}

\subsection{Greedy Algorithms}

\nop{\todo{This section is hard to follow.  It may be good if we can briefly describe the intuition of the major ideas informally before we jump into the detailed algorithms. Your algorithm names carry a prefix ``New''.  One may wonder what this ``New'' means.  You may want to either clarify or remove the word.}}

Although \textbf{DCSAD} cannot be approximated within $O(n^{1-\epsilon})$, an $O(n)$ approximation is easy to achieve. We have two cases,
\begin{enumerate}
    \item If there are no edges with positive weights in $G_D$, apparently any $S$ that only contains a single vertex is an optimal solution to the \textbf{DCSAD} problem, and the optimal density difference is 0.
    \item If $G_D$ has at least one edge with positive weight, $S=\{u,v\}$ is an $O(n)$ approximation solution, where $(u,v)=\arg\max_{(u,v) \in E_D}{D(u,v)}$. The reason is as follows. For any $S' \subseteq V$, $\rho_D(S')$ must be no greater than the density of an $n$-clique where every edge's weight is $D(u,v)$. Such an $n$-clique has density $(n-1)D(u,v)$. Note that $\rho_D(S)=D(u,v)$. Thus, $\frac{\rho_D(S)}{\max_{S' \subseteq V}{\rho_D(S')}} \geq n-1 = O(n)$.
\end{enumerate}

Utilizing the above results, and inspired by the greedy approximation algorithm (shown in Algorithm~\ref{alg:greedy}) for the traditional dense subragph discovery problem~\cite{charikar2000greedy}, we devise an $O(n)$ approximation algorithm, the DCSGreedy algorithm (Algorithm~\ref{alg:ad}), which also has \updates{a data-dependent} ratio.

The idea of the Algorithm~\ref{alg:ad} is to generate multiple potentially good solutions and pick the best one. As discussed above, when $G_D$ has positive weighted edges, the edge $(u,v)$ with the maximum weight is a candidate solution since it is $\frac{1}{n-1}$-optimal. The Greedy algorithm may also generate a good solution, although for the \textbf{DCSAD} problem its approximation ratio is no better than $O(n^{1-\epsilon})$ for any $\epsilon>0$. Thus, we run Algorithm~\ref{alg:greedy} on $G_D$ to generate $S_1$. We also run Algorithm~\ref{alg:greedy} on $G_{D^+}$ to get $S_2$, because not only $S_2$ may be a better solution, but also $\rho_{D^+}(S_2)$, the average degree of $S_2$ in $G_{D^+}$, helps us derive \updates{a data-dependent} ratio of Algorithm~\ref{alg:ad}, which will be shown in Theorem~\ref{th:online_ratio}. In line~\ref{line:cc} of Algorithm~\ref{alg:ad}, $CC_D(S)$ is the set of connected components of $G_D(S)$, where a connected component is represented by a set of vertices. Line~\ref{line:cc} is for refining the solution $S$ obtained at line~\ref{line:max} when $G_D(S)$ is not connected, since \textbf{DCSAD} prefers ``connected'' subgraphs.

\begin{theorem}\label{th:online_ratio}
	The $S$ returned by Algorithm~\ref{alg:ad} has \updates{a data-dependent} ratio of $\frac{2\rho_{D^+}(S_2)}{\rho_D(S)}$, where $S_2$ is the set in line~\ref{line:greedy2} of Algorithm~\ref{alg:ad}.
\end{theorem}
\begin{proof}
	It is known that $\rho_{D^+}(S_2)$ is a 2-approximation of the maximum density in $G_{D^+}$~\cite{charikar2000greedy}. For any $S' \subseteq V$, clearly $\rho_D(S') \leq \rho_{D^+}(S')$. Thus, the maximum density in $G_D$ is at most $2\rho_{D^+}(S_2)$ and the $S$ returned by Algorithm~\ref{alg:ad} has \updates{a data-dependent} ratio of $\frac{2\rho_{D^+}(S_2)}{\rho_D(S)}$.
\end{proof}

We analyze the time complexity of Algorithm~\ref{alg:ad}. Suppose $|V|=n$, $|E_1|=m_1$ and $|E_2|=m_2$. The difference graph $G_D$ can be built in $O((m_1+m_2)\log{n}+n)$ time, if we sort the adjacent lists of $G_1$ and $G_2$ first, and use then a merge sort to build $u$'s adjacent list in $G_D$ for each $u \in V$. Finding the maximum edge weight can be done in $O(m_1+m_2)$ time since $G_D$ has at most $m_1+m_2$ edges. Running the Greedy algorithm on a graph $G=\langle V,E,A \rangle$ can be finished in $O((|E|+|V|)\log{|V|})$ time, if we adopt a segment tree~\cite{bentley1977solutions} to store the current degrees of vertices in $S_1$. Thus, $\textbf{Greedy}(G_D)$ and $\textbf{Greedy}(G_{D^+})$ together can be done in $O((m_1+m_2+n)\log{n})$ time. Lines~\ref{line:cc_start} and~\ref{line:cc} obviously can be done in $O(m_1+m_2+n)$ time. Thus, in total Algorithm~\ref{alg:ad} can be efficiently implemented in $O((m_1+m_2+n)\log{n})$ time.

\begin{algorithm}[t]
\TitleOfAlgo{\textbf{Greedy}}
\caption{Greedy Algorithm.}
\label{alg:greedy}
\KwIn{$G=\langle V,E,A \rangle$}
\KwOut{$S$}
\begin{algorithmic}[1]
\STATE $S \leftarrow V$, $S_1 \leftarrow V$
\WHILE {$|S_1| \geq 1$}
    \IF {$\frac{W(S_1)}{|S_1|} > \frac{W(S)}{|S|}$}
        \STATE $S \leftarrow S_1$
    \ENDIF
    \STATE $i \leftarrow \arg\min_{j \in S_1}{W(j;G(S_1))}$ \label{line:pick}
    \STATE $S_1 \leftarrow S_1 \setminus \{i\}$
\ENDWHILE
\RETURN $S$
\end{algorithmic}
\end{algorithm}

\begin{algorithm}[t]
\TitleOfAlgo{\textbf{DCSGreedy}}
\caption{DCSGreedy algorithm for solving \textbf{DCSAD}.}
\label{alg:ad}
\KwIn{$G_1=\langle V,E_1,A_1 \rangle$, $G_2=\langle V,E_2,A_2 \rangle$}
\KwOut{$S$, and \updates{a data-dependent} ratio $\beta$}
\begin{algorithmic}[1]
\STATE Build the difference graph $G_D=\langle V,E_D \rangle$
\IF {$G_D$ does not have edges with positive weights}
    \STATE Randomly pick a vertex $v$
    \RETURN $S \leftarrow \{v\}$ 
\ENDIF
\STATE $(u,v) \leftarrow \arg\max_{(u,v) \in E_D}{D(u,v)}$
\STATE $S \leftarrow \{u,v\}$, $S_1 \leftarrow \textbf{Greedy}(G_D)$ \label{line:greedy1}, $S_2 \leftarrow \textbf{Greedy}(G_{D^+})$ \label{line:greedy2}
\STATE $S \leftarrow \arg\max_{S' \in \{S, S_1, S_2\}}\frac{W_D(S')}{|S'|}$ \label{line:max}
\IF {$G_D(S)$ is not connected}\label{line:cc_start}
    \STATE $S \leftarrow \arg\max_{S' \in CC_D(S)}{\frac{W_D(S')}{|S'|}}$ \label{line:cc} 
\ENDIF
\RETURN $S$ and $\beta \leftarrow \frac{2\rho_{D^+}(S_2)}{\rho_D(S)}$
\end{algorithmic}
\end{algorithm}

\section{DCS with respect to Graph Affinity}\label{sec:ga}
In this section, we first explore several properties of the \textbf{DCSGA} problem. Then, we devise a Coordinate-Descent algorithm which is guaranteed to converge to a KKT point. We also propose a refinement step to further improve a KKT point solution. Since \textbf{DCSGA} is non-concave, normally we need multiple initializations to find a good solution. To reduce the number of initializations, we utilize a smart initialization heuristic. Combining the Coordinate-Descent algorithm, the refinement step and the smart initialization heuristic together, we have our NewSEA algorithm for the \textbf{DCSGA} problem. 

\subsection{Properties}
We first show that, like the \textbf{DCSAD} problem, \textbf{DCSGA} also prefers connected subgraphs in the difference graph $G_D$.

\begin{property}
Let $G_D=\langle V,E_D,D \rangle$ be the difference graph of $G_1$ and $G_2$. For any $\textbf{x} \in \triangle^n$ such that $f_D(\textbf{x})=\textbf{x}^\top D \textbf{x} \geq 0$, if $G_D(S_{\textbf{x}})$ is not connected, where $S_{\textbf{x}}$ is the support set of $\textbf{x}$, then there exists $\textbf{x}'$ whose support set $S_{\textbf{x}'} \subseteq S_{\textbf{x}}$, and $G_D(S_{\textbf{x}'})$ is connected, and $f_D(\textbf{x}') \geq f_D(\textbf{x})$.  
\end{property}
\begin{proof}
Without loss of generality, we assume $G_D(S)$ has two connected components $G_D(S_1)$ and $G_D(S_2)$, where $S_1 \cup S_2=S$ and $S_1 \cap S_2=\emptyset$. We decompose $\textbf{x}$ such that $\textbf{x}=\textbf{y}+\textbf{z}$, where $S_{\textbf{x}}=S_1$ and $S_{\textbf{y}}=S_2$. Because $S_1$ and $S_2$ are two connected components in $G_D(S)$, we have $\textbf{y}^\top D \textbf{z}=0$. Thus, $\textbf{x}^\top D \textbf{x}=(\textbf{y}+\textbf{z})^\top D (\textbf{y}+\textbf{z}) = \textbf{y}^\top D \textbf{y} + \textbf{z}^\top D \textbf{z}$. Let $\textbf{y}'=\frac{\textbf{y}}{|\textbf{y}|_1}$ and $\textbf{z}'=\frac{\textbf{z}}{|\textbf{z}|_1}$. Clearly, $\textbf{y}' \in \triangle^n$ and $\textbf{z}' \in \triangle^n$. So both $\textbf{y}'$ and $\textbf{z}'$ are subgraph embeddings. Thus, we have$f_D(\textbf{x})=\textbf{x}^\top D \textbf{x}=|\textbf{y}|_1^2f_D(\textbf{y}')+|\textbf{z}|_1^2f_D(\textbf{z}')$. Since $\textbf{x}^\top D \textbf{x} \geq 0$ and both $|\textbf{y}|_1^2$ and $|\textbf{z}|_1^2$ are non-negative, $\max\{f_D(\textbf{y}'),f_D(\textbf{z}')\} \geq 0$. Also $|\textbf{y}|_1+|\textbf{z}|_1=|\textbf{x}|_1=1$, so $|\textbf{y}|_1^2+|\textbf{z}|_1^2 \leq 1$. We get that $f_D(\textbf{x}) \leq (|\textbf{y}|_1^2+|\textbf{z}|_1^2)\max\{f_D(\textbf{y}'),f_D(\textbf{z}')\} \leq \max\{f_D(\textbf{y}'),f_D(\textbf{z}')\}$.
\end{proof}

The \textbf{DCSGA} is a standard Quadratic Programming (QP) problem, which in general is NP-hard. We prove that \textbf{DCSGA} is NP-hard. 

\begin{theorem}
The \textbf{DCSGA} (Eq.~\ref{eq:DCS_ga}) problem is NP-hard.
\end{theorem}
\begin{proof}
Consider an undirected and unweighted graph $G$ whose adjacency matrix is $A$, where the entries of $A$ are either 0 or 1. It is known that maximizing $\textbf{x}^\top A \textbf{x}$ s.t. $\textbf{x} \in \triangle^n$ is NP-hard, because the optimum is $1-\frac{1}{k}$, where $k$ is the size of the maximum clique of $G$~\cite{motzkin1965maxima}. Given an arbitrary undirected and unweighted graph $G$, we create a corresponding instance of the \textbf{DCSGA} problem by building $G_1$ as a graph without any edges and setting $G_2=G$. Clearly for any $\textbf{x} \in \triangle^n$, we have $\textbf{x}^\top A \textbf{x}=\textbf{x}^\top D \textbf{x}$, where $D$ is the affinity matrix of the difference graph between $G_2$ and $G_1$. Thus, this simple reduction proves that the \textbf{DCSGA} problem is also NP-hard.
\end{proof}

\subsection{The SEACD Algorithm}\label{sec:seacd}
Since $\textbf{DCSGA}$ is NP-hard and is a QP, we employ local search algorithms to find good solutions. Because the density difference $\textbf{x}^\top D \textbf{x}$ is normally non-concave, we seek for $\textbf{x}$ that satisfies the Karush-Kuhn-Tucker (KKT) conditions~\cite{boyd2004convex}, which are necessary conditions of local maxima points. It is easy to derive that, if $\textbf{x}$ is a KKT point of the \textbf{DCSGA} problem, it should satisfy
\begin{equation}\label{eq:KKT}
    \nabla_uf_D(\textbf{x})=2(D\textbf{x})_u
        \begin{cases}
            =\lambda~~x_u>0 \\
            \leq \lambda~~x_u=0
        \end{cases}
    \forall u \in V
\end{equation}
where $\nabla_uf_D(\textbf{x}^*)$ is the partial derivative with respect to $x_u$, and $(D\textbf{x}^*)_u$ is the $u$-th entry of the vector $D\textbf{x}^*$. Since $x \in \triangle^n$, when Eq.~\ref{eq:KKT} holds, we have $f_D(\textbf{x})=\sum_{u \in V}{x_u*(D\textbf{x})_u}=\frac{\lambda}{2}$. 

The condition in Eq.~\ref{eq:KKT} is also equivalent to 
\begin{equation}\label{eq:KKT_eq}
    \max_{k:x_k<1}{\nabla_kf_D(\textbf{x})} \leq \min_{k:x_k>0}{\nabla_kf_D(\textbf{x})}
\end{equation}

The Shrink-and-Expansion (SEA\footnote{The details of SEA algorithm are illustrated in Appendix.}) algorithm in~\cite{liu2013fast} utilizes a replicator dynamic to solve the problem that maximizes $\textbf{x}^\top A \textbf{x}$ s.t. $\textbf{x} \in \triangle^n$, where $A$ is an affinity matrix of an undirected graph. Although $D$ in Eq.~\ref{eq:DCS_ga} can also be regarded as an affinity matrix, unfortunately the SEA algorithm cannot be directly applied to our problem. This is because the replicator dynamic can only deal with non-negative matrices, while in our problem the matrix $D$ may have negative entries. Thus, we devise a 2-Coordinate Descent algorithm to solve Eq.~\ref{eq:DCS_ga}.

In every iteration of the 2-Coordinate Descent algorithm, we only pick two variables $x_i$ and $x_j$, and fix the rest $n-2$ variables. We adjust the values of $x_i$ and $x_j$ to increase the objective $f_D(\textbf{x})$ without violating the simplex constraint. Suppose $x_i+x_j=C$, and let $b_i=\sum_{a \in N_D(i),a \neq j}{D(a,i)x_a}$, $b_j=\sum_{a \in N_D(j),a \neq i}{D(a,j)x_a}$, where $N_D(i)$ is the set of $i$'s neighbors in $G_D$. We adjust $x_i$ and $x_j$ by solving a simple optimization problem involving only one variable, since $x_j$ should always equal $C-x_i$ when the rest $n-2$ variables are fixed. Specifically, the optimization problem is
\begin{equation}\label{eq:CD}
    \begin{split}
        &\text{max}~~~\frac{1}{2}f_D(\textbf{x})=g(x_i)=b_ix_i+b_j(C-x_i)+D(i,j)x_i(C-x_i)+Cnst \\
        &\text{s.t.}~~~~0 \leq x_i \leq C
    \end{split}
\end{equation}
where $Cnst$ is a constant independent from $x_i$ and $x_j$.

Eq.~\ref{eq:CD} can be solved analytically. There are two cases,
\begin{enumerate}
    \item $D(i,j)=0$, which means $i$ and $j$ are not adjacent in the difference graph $G_D$. Then $g(x_i)=(b_i-b_j)x_i+b_jC+Cnst$. Obviously we should set $x_i=C$ if $b_i > b_j$, and set $x_i=0$ if $b_i < b_j$. We do not adjust $b_i$ or $b_j$ if $b_i=b_j$.
    \item $D(i,j) \neq 0$, which means $i$ and $j$ are adjacent in $G_D$. We have $g(x_i)=-D(i,j)x_i^2+Bx_i+b_jC+Cnst$ where $B=D(i,j)C+b_i-b_j$. Let $r=\frac{B}{2D(i,j)}$. If $0 \leq r \leq C$, we set $x_i=\arg\max_{x \in \{0,r,c\}}{g(x)}$. If $r<0$ or $r>C$, we set $x_i=\arg\max_{x \in \{0,C\}}{g(x)}$.
\end{enumerate}

To pick $x_i$ and $x_j$ for an iteration, we exploit the partial derivatives. We pick $i=\arg\max_{k:x_k<1}{\nabla_kf_D(\textbf{x})}$ and $j=\arg\min_{k:x_k>0}{\nabla_kf_D(\textbf{x})}$. If $\nabla_if_D(\textbf{x}) \leq \nabla_jf_D(\textbf{x})$, which means we reach a KKT point, the algorithm stops.

The 2-Coordinate Descent algorithm is guaranteed to converge to a stationary point, which is equivalent to a KKT point because the constraint $\textbf{x} \in \triangle^n$ in Eq.~\ref{eq:DCS_ga} is linear~\cite{boyd2004convex}.

Picking $x_i$ and $x_j$ at the beginning of every iteration clearly can be done in $O(n)$ time. But $O(n)$ may still be too costly for large graphs. Thus, to further improve the efficiency of our algorithm, we adopt the strategy of the Shrink-and-Expansion algorithm. We define a \textbf{local KKT point} on $S \subseteq V$ as a point $\textbf{x} \in \triangle^n$ that satisfies the following conditions,
\begin{equation}\label{eq:KKT_local}
    \begin{split}
    &x_u=0~~if~~u \notin S \\
    &\nabla_uf_D(\textbf{x})=2(D\textbf{x})_u
        \begin{cases}
            =\lambda~~x_u>0 \\
            \leq \lambda~~x_u=0\\
        \end{cases}
    \forall u \in S \\
    & \lambda=2f_D(\textbf{x})
    \end{split}
\end{equation}
where the major difference from Eq.~\ref{eq:KKT} is that only the vertices in $S \subseteq V$ are considered. It is also equivalent to 
\begin{equation}\label{eq:KKT_local_eq}
	\max_{k \in S:x_k<1}{\nabla_kf_D(\textbf{x})} \leq \min_{k \in S:x_k>0}{\nabla_kf_D(\textbf{x})}
\end{equation}
The 2-Coordinate Descent algorithm is guaranteed to converge to a local KKT point on $S$, when we keep $x_u=0$ for every $u \notin S$, and $x_u$ is involved in iterations only when $u \in S$.

Algorithm~\ref{alg:seacd} shows our method. We start with an initial embedding $\textbf{x} \in \triangle^n$. Line~\ref{line:shrink} is the Shrink stage, since after calling the 2-coordinate descent algorithm, the support set of $\textbf{x}$ may shrink due to some originally positive $x_i$ is set to 0. Line~\ref{line:expansion} is the start of the expansion stage. We first enlarge $S$ by adding to $S$ the vertices whose partial derivatives are greater than $\lambda=2f_D(\textbf{x})$, and then do exactly the same expansion operation of the original SEA algorithm~\cite{liu2013fast} (see Appendix). If $Z$ in Line~\ref{line:expansion} is empty, the current $\textbf{x}$ is already a KKT point satisfying conditions in Eq.~\ref{eq:KKT} and the SEA iterations stop. 

\begin{algorithm}[t]
\TitleOfAlgo{\textbf{SEACD}}
\caption{Coordinate Descent SEA Algorithm.}
\label{alg:seacd}
\KwIn{$G_D$, an initial embedding $\textbf{x} \in \triangle^n$}
\KwOut{$\textbf{x}$}
\begin{algorithmic}
    \STATE $S \leftarrow S_{\textbf{x}}$
    \WHILE {\textbf{true}}
        \STATE Use the 2-Coordinate Descent algorithm and take $\textbf{x}$ as the initial value to find a local KKT point $\textbf{x}^{new}$ on $S$  \label{line:shrink}
        \STATE $\textbf{x} \leftarrow \textbf{x}^{new}$
        \STATE $S \leftarrow \{v \mid x_v > 0\}$, $\lambda \leftarrow 2f_D(\textbf{x})$
        \STATE $Z \leftarrow \{i \mid \nabla_if_D(\textbf{x}) > \lambda, i \in V\}$  \label{line:expansion}
        \IF {$Z=\emptyset$}
        	   \STATE \textbf{break} 
        \ENDIF
        \STATE Do the SEA Expansion operation on $S \cup Z$ to adjust $\textbf{x}$ \label{line:sea_expansion}
         \STATE $S \leftarrow S_{\textbf{x}}$
    \ENDWHILE \label{line:sea_end} 
\RETURN $\textbf{x}$
\end{algorithmic}
\end{algorithm}

Like the original SEA algorithm~\cite{liu2013fast}, The SEACD algorithm converges to a KKT point.
\begin{theorem}\label{th:converge}
The SEACD algorithm (Algorithm~\ref{alg:seacd}) is guaranteed to converge to a KKT point.
\end{theorem}

We analyze the computational cost of Algorithm~\ref{alg:seacd}. It is worth noting that to efficiently run Algorithm~\ref{alg:seacd}, the initial embedding $\textbf{x}$ should have a small support set such that during the execution of Algorithm~\ref{alg:seacd}, $S$ and $Z$ in the while loop are normally small sets. In the Shrink stage, for every iteration, we need $O(|S|)$ time to pick $x_i$ and $x_j$, $O(1)$ time to adjust $x_i$ and $x_j$, and $O(|N_D(i)|+|N_D(j)|)$ time to update the partial derivatives of the vertices affected by adjusting $x_i$ or $x_j$. $S$ is usually a small set and $|N_D(i)|+|N_D(j)|$ is often a small number since real-world graphs are normally sparse. Thus, the cost of each iteration of the shrink stage is low. In Line~\ref{line:expansion} of the Expansion stage, we only need to check the partial derivatives of the vertices that have at least one neighbor in $S$, since the partial derivatives of all other vertices are $0$. Thus, the cost of Line~\ref{line:expansion} is $\sum_{v \in S}{|N_D(v)|}$. Line~\ref{line:sea_expansion} is the same as the Expansion operation of the SEA algorithm, whose cost is $O(\sum_{v \in S \cup Z}{|N_D(v)|})$~\cite{liu2013fast}. Since both $S$ and $Z$ are normally small sets, the cost of one SEA iteration (one Shrink stage + one Expansion stage) is low.

\subsection{Refining a KKT Point Solution}\label{sec:refine}
After a KKT point solution is reached, we may further improve the solution. We call a clique in $G_D$ as a \textbf{positive clique} if all its edge weights are positive, and $\textbf{x} \in \triangle^n$ a \textbf{positive clique solution} if $G_D(S_{\textbf{x}})$ is a positive clique. Utilizing the 2-Coordinate Decent algorithm, we give a construction that refines a KKT point $\textbf{x}$ such that $G(S_{\textbf{x}})$ is not a positive clique to a better solution.

\begin{theorem}\label{th:clique}
For any KKT point $\textbf{x}$, let $S_{\textbf{x}}=\{v \mid x_v>0\}$. If $G_D(S_{\textbf{x}})$ is not a positive clique, we can find a $\textbf{y}$ such that $G_D(S_{\textbf{y}})$ is a positive clique and $f_D(\textbf{y}) \geq f_D(\textbf{x})$, where $S_{\textbf{y}}=\{v \mid y_v>0\}$ and $S_{\textbf{y}} \subseteq S_{\textbf{x}}$.
\end{theorem}
\begin{proof}
Suppose $\textbf{x}$ is a KKT point and $G_D(S_{\textbf{x}})$ is not a positive clique. We pick $x_i$ and $x_j$ from $S_{\textbf{x}}$ such that $D(i,j) \leq 0$. 

If $D(i,j)=0$, since $\nabla_if_D(\textbf{x})=\nabla_jf_D(\textbf{x})$, we have $(D\textbf{x})_i=(D\textbf{x})_j$ which means $\sum_{a \in N_D(i)}{D(a,i)x_a}=\sum_{a \in N_D(j)}{D(a,j)x_a}$. Thus, $b_i=b_j$ in Eq.~\ref{eq:CD} and we have $g(x_i)=b_jC+Cnst$. Note that $b_jC$ is independent of $x_i$ and $x_j$, as long as $x_i+x_j=C$. We set $x_i=C$ and $x_j=0$ to remove vertex $j$ from the current subgraph, and the objective $f_D(\textbf{x})$ remains the same.

If $D(i,j)<0$, we solve the optimization problem in Eq.~\ref{eq:CD}. Apparently $g(x_i)$ is a convex function with respect to $x_i$ because $-D(i,j)>0$. To maximize the objective $g(x_i)$, we should set $x_i^{new}=\arg\max_{x \in \{0,C\}}{g(x)}$. Thus, after solving Eq.~\ref{eq:CD}, either $x_i$ or $x_j$ becomes 0 and the objective $f_D(\textbf{x})$ is improved.

Thus, if $G_D(S_{\textbf{x}})$ is not a positive clique, we can always remove one vertex $i$ (by setting $x_i=0$) that is incident to an edge with negative weight or is not adjacent to all other vertices in $S_{\textbf{x}}$, and keep the objective non-decreasing. Suppose after removing this vertex we get $\textbf{y}$. We use the 2-coordinate descent algorithm to adjust $\textbf{y}$ to a local KKT point on $S_{\textbf{y}}$, and obviously the objective $f_D(\textbf{y})$ is not decreased. If $G_D(S_{\textbf{y}})$ is still not a positive clique, we repeat the above procedure of removing one vertex and adjusting to a local KKT point. During this process, the support set shrinks if the current solution is not a positive clique solution. Since the support set cannot shrink forever (it should has at least 1 vertex), finally we will reach a positive clique solution $\textbf{y}$ such that $S_{\textbf{y}} \subseteq S_{\textbf{x}}$. Moreover, during the process of reaching $\textbf{y}$, the objective is non-decreasing. Thus, we have $f_D(\textbf{y}) \geq f_D(\textbf{x})$.
\nop{
During this process, the support set of our solution keeps shrinking, and the objective is non-decreasing. Thus, finally we will reach a solution $\textbf{y}$ such that $f_D(\textbf{y}) \geq f_D(\textbf{x})$ and $G_D(S_{\textbf{y}})$ is a positive clique. \todo{You need to specifically explain why $G_D(S_{\textbf{y}})$ will ultimately turn to positive.} At the same time, $S_{\textbf{y}} \subseteq S_{\textbf{x}}$.}
\end{proof}

Since an optimal $\textbf{x}$ must be a KKT point, Theorem~\ref{th:clique} implies that there exist a solution $\textbf{y} \in \triangle^n$ such that $\textbf{y}$ is an optimal solution to Eq.~\ref{eq:DCS_ga}, and $G_D(S_{\textbf{y}})$ is a positive clique in $G_D$. Note that a positive clique in $G_D$ is a clique in $G_{D^+}$. Thus, we can run Algorithm~\ref{alg:seacd} directly on $G_{D^+}$ instead of $G_D$ to get a solution $\textbf{x}$. If $G_{D^+}(S_{\textbf{x}})$ is not a clique in $G_{D^+}$, we use the construction in the proof of Theorem~\ref{th:clique} to find a new solution $\textbf{y}$ whose $G_{D^+}(S_{\textbf{y}})$ is a clique. Algorithm~\ref{alg:refine} shows the construction, where we do not consider the case when $D^+(i,j)<0$ since $D^+$ only has non-negative entries.

Since the edges with negative weights can be ignored, it seems we can run the original SEA algorithm~\cite{liu2013fast} on $G_{D^+}$ directly to find DCS. However, SEA in~\cite{liu2013fast} is not guaranteed to return a positive clique solution $\textbf{x}$. If $G_{D^+}(S_{\textbf{x}})$ is not a clique, $G_{D}(S_{\textbf{x}})$ may have some edges with negative weights and $\textbf{x}$ is definitely not an optimal solution. This is because according to the proof of Theorem~\ref{th:clique}, if $D(i,j)<0$ where $x_i>0$ and $x_j>0$, we can solve the optimization problem in Eq.~\ref{eq:CD} over $x_i$ and $x_j$ to further improve the objective. Therefore, we still need our refinement step (Algorithm~\ref{alg:refine}).

\begin{algorithm}[t]
\TitleOfAlgo{\textbf{Refinement}}
\caption{Refining a KKT point.}
\label{alg:refine}
\KwIn{$G_{D^+}$, a KKT point $\textbf{x}$}
\KwOut{$\textbf{y}$}
\begin{algorithmic}[1]
\STATE $\textbf{y} \leftarrow \textbf{x}$
\WHILE {$G_{D^+}(S_\textbf{y})$ is not a clique}
    \STATE Pick $u$ and $v$ such that $(u,v)$ is not an edge in $G_{D^+}$
    \STATE $y_u \leftarrow y_u+y_v$, $y_v \leftarrow 0$
    \STATE Use the 2-Coordinate Descent algorithm and take $\textbf{y}$ as the initial value to find a local KKT point $\textbf{y}^{new}$ on $S_{\textbf{y}}$
    \STATE $\textbf{y} \leftarrow \textbf{y}^{new}$
\ENDWHILE
\RETURN $\textbf{y}$
\end{algorithmic}
\end{algorithm}

Always returning a positive clique solution as the DCS is one advantage of adopting graph affinity as the density measure, since the returned DCS has very good interpretability. From $G_1$ to $G_2$, for every pair of vertices in the DCS, their connection is enhanced.

Please note that, although there exists an optimal solution $\textbf{x}$ such that $G_D(S_{\textbf{x}})$ is a positive clique, it does not mean a maximum clique finding algorithms like~\cite{rossi2014fast} can be applied to solve the \textbf{DCSGA} problem. The major reason is that $G_D$ in the \textbf{DCSGA} problem is a weighted graph while maximum clique finding algorithms deal with unweighted graphs.

\noindent \textbf{Advantages of the Coordinate-Descent SEA} With the help of the refinement step (Algorithm~\ref{alg:refine}), the original SEA algorithm~\cite{liu2013fast} works for the DCSGA problem. However, our Coordinate-Descent SEA algorithm has some advantages over the original SEA algorithm. The correctness of the Expansion operation (see Appendix) depends on that a local KKT point is reached in the Shrink stage. Thus, when implementing the Shrink stage, the correct condition of convergence should be $\max_{k \in S:x_k<1}{\nabla_kf_D(\textbf{x})}-\min_{k \in S:x_k>0}{\nabla_kf_D(\textbf{x})} \leq \epsilon$, where $S$ is the set of vertices on which we try to find a local KKT point, and $\epsilon$ is the parameter of precision. However, the original SEA~\cite{liu2013fast} adopts $f_D(\textbf{x})-f_D(\textbf{x}^{old}) \leq \epsilon$ as the convergence condition, where $\textbf{x}$ and $\textbf{x}^{old}$ are the solutions after and before a Shrink iteration by the replicator dynamics. In Section~\ref{sec:exp}, we will show that such a convergence condition may fail to achieve a local KKT point and as a result, the objective $f_D(\textbf{x})$ is even reduced in the following Expansion stage. Moreover, when the convergence condition of the Shrink stage is correctly set, the replicator dynamics of the original SEA~\cite{liu2013fast} converges much slower than the coordinate-descent method, especially on dense graphs. Since our algorithm can also deal with graph with only positive edge weights, it is also a competitive solution to the traditional graph affinity maximization problem.

\subsection{Smart Initializations of $\textbf{x}$}
One problem remaining unsettled is how to choose the initial embedding $\textbf{x}$ for running Algorithm~\ref{alg:seacd}. Since the \textbf{DCSGA} problem is non-concave, we adopt the strategy of multiple initializations, that is, we run Algorithm~\ref{alg:seacd} multiple times with different initial embeddings. The best solution generated in all runs is returned as the final solution. For the interest of efficiency of the SEACD algorithm as illustrated in Section~\ref{sec:seacd}, an initial embedding $\textbf{x}$ should have a small support set.

One simple way of initialization is to set $\textbf{x}=\textbf{e}_u$, where in $\textbf{e}_u$, only the $u$-th entry is 1 and all other entries are 0. The original SEA algorithm employs this simple method and it uses every vertex $u \in V$ to set the initial embedding~\cite{liu2013fast}. Thus, in~\cite{liu2013fast}, the SEA algorithm is called $n=|V|$ times. 

For large graphs, $O(n)$ initializations are clearly very time-consuming. We adopt a smart heuristic to reduce the number of initializations. The major idea is to first find an upper bound $\mu_u$ for each $u \in V$, where $\mu_u$ is the upper bound of $\textbf{x}^{\top}D\textbf{x}$ for any $\textbf{x} \in \triangle^n$ such that $x_u>0$ and $G_{D^+}(S_{\textbf{x}})$ is a clique. Then we only use the vertices with big upper bounds to do initializations.

Define the \textbf{ego net} of $u$ in $G_{D^+}$ as $G_{D^+}(T_u)$ where $T_u$ is the set containing $u$ and all $u$'s neighbors in $G_{D^+}$. Let $w_u=\max_{i \in T_u \vee j \in T_u}{D^+(i,j)}$. Clearly, $w_u$ is an upper bound of the maximum edge weight in $u$'s ego net. Using $O(|E_{D^+}|)$ time, we compute $w_u$ for every $u \in V$. 
\begin{theorem}\label{th:bound}
For any $u \in V$, $\textbf{x}^{\top}D\textbf{x} \leq \frac{(k-1)w_u}{k}$, where $\textbf{x} \in \triangle^n$ and $G_{D^+}(S_{\textbf{x}})$ is a $k$-clique containing $u$, and $w_u$ is an upper bound of the maximum edge weight in $G_{D^+}(T_u)$, the ego net of $u$ in $G_{D^+}$.
\end{theorem}
\begin{proof}
Suppose for $\textbf{x} \in \triangle^n$, $G_{D^+}(S_{\textbf{x}})$ is a $k$-clique containing $u$. Thus, $\textbf{x}^{\top}D^+\textbf{x}=\sum_{(i,j) \in E_{D^+}(S_{\textbf{x}})}{x_ix_jD^+(i,j)}=\textbf{x}^{\top}D\textbf{x}$. Since $G_{D^+}(S_{\textbf{x}})$ is a $k$-clique containing $u$, for any $(i,j) \in E_{D^+}(S_{\textbf{x}})$, $D(i,j) \leq w_u$. Therefore,  $\textbf{x}^{\top}D\textbf{x} \leq w_u\sum_{(i,j) \in E_{D^+}(S_{\textbf{x}})}{x_ix_j}$. When $G_{D^+}(S_{\textbf{x}})$ is a $k$-clique, it is easy to find that $\sum_{(i,j) \in E_{D^+}(S_{\textbf{x}})}{x_ix_j} \leq \sum_{(i,j) \in E_{D^+}(S_{\textbf{x}})}{\frac{1}{k}\frac{1}{k}}=\frac{k-1}{k}$. Thus, $\textbf{x}^{\top}D\textbf{x} \leq \frac{(k-1)w_u}{k}$.
\end{proof}

Based on Theorem~\ref{th:bound}, assuming $k_u$ is the size of the maximum clique in $G_{D^+}$ that contains $u$, then $\textbf{x}^TD^{+}\textbf{x}$ is no more than $\frac{(k_u-1)w_u}{k_u}$, where $\textbf{x} \in \triangle^n$, $x_u>0$ and $G_{D^+}(S_{\textbf{x}})$ is a clique. Although computing $k_u$ for every $u \in V$ is NP-hard, it is easy to find an upper bound of $k_u$, which is $\tau_u+1$ where $\tau_u$ is the core number of $u$ in $G_{D^+}$~\cite{rossi2014fast}. Thus, we use $\mu_u=\frac{\tau_uw_u}{\tau_u+1}$ as the upper bound of the affinity of a clique in $G_{D^+}$ that contains $u$. Note that computing $\tau_u$ for every $u \in V$ can be done in $O(|E_{D^+}|)$ time~\cite{rossi2014fast}.
  
We sort all vertices in $V$ in the descending order of $\mu_u$. Then we use the new order of vertices to initialize $\textbf{x}$. Suppose we have tried some vertices and $\textbf{y}$ is the current best solution. Then all vertices $v$ such that $\mu_v \leq f_D(\textbf{y})$ will not be used to initialize $\textbf{x}$. In such a case, normally we only need to do a small number of initializations. 

Note that when we use a vertex $u$ to initialize $\textbf{x}$, it is not guaranteed that after running the SEACD algorithm and the Refinement algorithm a solution $\textbf{x}$ is returned where $x_u>0$. It is possible that $x_u=0$ in the returned solution $\textbf{x}$. Thus, our method for reducing the number of initializations is not a pruning technique, but a heuristic. In Section~\ref{sec:exp} we show that our smart initialization heuristic is very effective and it never impairs the quality of the final solution $\textbf{x}$ compared to trying all vertices for initializations in experiments.

\medskip
Combining all results in this section, we propose the NewSEA algorithm shown in Algorithm~\ref{alg:new_sea}.
\begin{algorithm}[t]
\TitleOfAlgo{\textbf{NewSEA}}
\caption{The NewSEA algorithm for solving \textbf{DCSGA}.}
\label{alg:new_sea}
\KwIn{$G_{D^+}$}
\KwOut{$\textbf{y}$}
\begin{algorithmic}[1]
\STATE $\textbf{y} \leftarrow \textbf{0}$
\STATE Compute $w_u$, $\tau_u$ for every $u \in V$
\STATE Compute $\mu_u=\frac{\tau_uw_u}{\tau_u+1}$ for every $u \in V$
\STATE Sort $V$ in descending order of $\mu_u$
\FOR {$u \in V$}
	\IF {$\mu_u \leq f_D(\textbf{y})$}
		\STATE \textbf{break}
	\ENDIF
	\STATE Set $\textbf{x}$ such that $\textbf{x}_u=1$ and $\textbf{x}_v=0$ for all $v \neq u$
	\STATE $\textbf{x} \leftarrow \textbf{SEACD}(G_{D^+},\textbf{x})$
	\STATE $\textbf{x} \leftarrow \textbf{Refinement}(G_{D^+},\textbf{x})$
	\IF {$f_D(\textbf{x})> f_D(\textbf{y})$}
        		\STATE $\textbf{y} \leftarrow \textbf{x}$
    	\ENDIF
\ENDFOR
\RETURN $\textbf{y}$
\end{algorithmic}
\end{algorithm}

\section{Experiments}\label{sec:exp}
In this section, we report a series of experiments to verify the effectiveness and efficiency of our algorithms. 

\subsection{Algorithms and Datasets in Experiments}
For the \textbf{DCSAD} problem, we tested our DCSGreedy algorithm, the Greedy algorithm on $G_D$ (denoted by \textbf{$G_D$ only}) and the Greedy algorithm on $G_{D^+}$ (denoted by \textbf{$G_{D^+}$ only}). For the \textbf{DCSGA} problem, we tested our NewSEA algorithm, our SEACD algorithm plus the Refinement step but without our smart initializations heuristic (denoted by \textbf{SEACD+Refine}), and the original SEA algorithm~\cite{liu2013fast} plus the Refinement step (denoted by \textbf{SEA+Refine}). All \textbf{DCSGA} algorithms were run on $G_{D^+}$ directly. The convergence condition of the Shrink stage in NewSEA and SEACD+Refine are all set to $\max_{k \in S:x_k<1}{\nabla_kf_D(\textbf{x})}-\min_{k \in S:x_k>0}{\nabla_kf_D(\textbf{x})} \leq 10^{-2}*\frac{1}{|S|}$, where $S$ is the current set on which we want to reach a local KKT point. This convergence condition very often is too difficult to achieve for the replicator dynamic in the Shrink stage of SEA+Refine, because the replicator dynamic converges too slowly. Thus, for the Shrink stage in SEA+Refine, the convergence condition was set to that the improvement of the objective $f_D(\textbf{x})$ is less than $10^{-6}$ after one iteration. As pointed out in Section~\ref{sec:refine}, this convergence condition actually is not enough for achieving a local KKT point. Thus, in our experiments, the Shrink stage of SEA+Refine sometimes could not converge to a local KKT point and as a result, in the following Expansion stage error occurred, the objective $f_D(\textbf{x})$ was even reduced after expansion.  

We list the statistics of all data sets used in our experiments in Table~\ref{tab:gd_sta}.  The setting of ``Weighted'' represents that we built $G_D$ as $G_2-G_1$ directly. In some graphs there are several edges with weights significantly greater than the weights of other edges, they make the DCS with respect to graph affinity a very small subgraph, sometimes even a single edge. Thus, to limit the influence of these small number of edges with too heavy weights, we also tried the Discrete setting, where we set edge weights in $G_D$ discrete values such that the maximum weight is not too much greater than the other edge weights. Details of how to set edge weight in the Discrete Setting and ``$G_D$ Type" are illustrated in each task in the rest of this section.    
\begin{table*}
\centering
\begin{tabular}{|c|c|c|c|c|c|c|c|c|}
\hline
Data & Setting & $G_D$ Type & n & $m^+$ & $m^-$ & Max $w$ & Min $w$ & Average $w$ \\ \hline
DBLP & Weighted & Emerging & 22,572 & 61,703 & 61,551 & 46 & -100 & -0.015 \\ \hline
DBLP & Weighted & Disappearing & 22,572 & 61,551 & 61,703 & 100 & -46 & 0.015 \\ \hline
DBLP & Discrete & Emerging & 22,572 & 21,367 & 61,551 & 2 & -2 & -0.518 \\ \hline
DBLP & Discrete & Disappearing & 22,572 & 61,551 & 21,367 & 2 & -2 & 0.518 \\ \hline
DM & \NA & Emerging & 9890 & 140,705 & 67,541 & 1.988 & -5.997 & 0.0007 \\ \hline
DM & \NA & Disappearing & 9890 & 67,541 & 140,705 & 5.997 & -1.988 & -0.0007 \\ \hline
Wiki & \NA & Consistent &116,836 & 762,999 & 1,264,872 & 9.619 & -12.46 & -0.474 \\ \hline
Wiki & \NA & Conflicting &116,836 & 1,264,872 & 762,999 & 12.46 & -9.619 & 0.474 \\ \hline
Movie & \NA & Interest$-$Social & 55,710 & 338,524 & 914,292 & 1 & -1 & -0.46 \\ \hline
Movie & \NA & Socia$-$Interest & 55,710 & 914,292 & 338,524 & 1 & -1 & 0.46 \\ \hline
Book  & \NA & Interest$-$Social &55,710 & 124,027 & 918,925 & 1 & -1 & -0.762\\ \hline
Book  & \NA & Social$-$Interest &55,710 & 918,925 & 124,027 & 1 & -1 & 0.762\\ \hline
DBLP-C & Weighted & \NA & 1,282,461 & 2,538,746 & 2,359,487 & 400 & -186 & 0.188 \\ \hline
DBLP-C & Discrete & \NA & 1,282,461 & 2,538,746 & 2,359,487 & 2 & -2 & -0.013 \\ \hline
Actor  & Weighted & \NA & 382,219 & 15,038,083 & 0 & 216 & 1 & 1.101\\ \hline
Actor  & Discrete & \NA & 382,219 & 15,038,083 & 0 & 10 & 1 & 1.098\\ \hline
\end{tabular}
\caption{Statistics Difference Graphs in Experiments (\rm{$n$ represents \#vertices, $m^+$ is \#edges with positive weights and $m^-$ is \#edges with negative weights. ``Max $w$'' is the maximum edge weight while ``Min $w$'' is the minimum one. We also report the average edge weight in the column of ``Average $w$''. ``Setting'' and ``$G_D$ Type'' denote how the difference graph was built. ``$G_D$ Type'' denotes which graph is used as $G_1$ and which is used as $G_2$.} )}
\label{tab:gd_sta}
\end{table*}

\subsection{Finding Emerging and Disappearing Co-author Groups}

We applied DCS to find emerging/disappearing co-author groups from co-author networks. We adopted the DBLP dataset (\url{https://static.aminer.org/lab-datasets/citation/dblp.v8.tgz}) and extracted all papers published in the top conferences according to the CS Ranking website (\url{http://csrankings.org/}). Based on these papers, we built two co-author graphs. The first graph $G_1=\langle V,E_1,A_1 \rangle$ contains the co-authorships before the year of 2010, and the second one $G_2=\langle V,E_2,A_2 \rangle$ contains the co-authorships from 2010 to 2016. For an edge linking two authors in a co-author graph, the weight is the number of papers written by these two authors together. 

To build the difference graph $G_D=\langle V,E_D,D \rangle$, we tried two settings, the Weighted setting and the Discrete setting. In the Weighted setting, we set $D(u,v)=A_2(u,v)-A_1(u,v)$, which is the standard setting of the DCS problem. In the Discrete setting, the entries of $D$ are set to discrete values. Specifically, if $A_2(u,v)-A_1(u,v) \geq 5$, which means $u$ and $v$ have at least 5 more co-authored papers in $G_2$ than in $G_1$, we set $D(u,v)=2$. If $2 \leq A_2(u,v)-A_1(u,v)<5$ , we set $D(u,v)=1$. If $-4<A_2(u,v)-A_1(u,v)<0$, we set $D(u,v)=-1$. If $A_2(u,v)-A_1(u,v) \leq -4$, we set $D(u,v)=-2$. The two different settings of $G_D$ normally lead to different DCS.

Running our DCS algorithms on $G_D$ described above, no matter in Weighted setting or Discrete Setting, what we find is the Emerging co-author group whose strength (density) of collaborations was enhanced after 2010. Thus, the type of $G_D$ described above is called Emerging. We also wanted to mine the disappearing co-author group whose collaboration strength was weakened the most after 2010. Therefore, we tried another type of $G_D$, the Disappearing $G_D$, which was obtained by flipping the sign of weight of each edge in the Emerging $G_D$.

It turned out that, under the same $G_D$ and the same density measure, all algorithms find the same group of authors. We list all co-author groups obtained in Table~\ref{tab:group}. If a group is found under the graph affinity measure, the weight (in the simplex) of each author is also given. We give a short note on the affiliation/address and research interest of each group. Table~\ref{tab:group_info} reports the groups found under different settings and density measures. For the average degree measure, we also report the approximation ratio $\frac{2\rho_{D^+}(S_2)}{\rho_D(S)}$. For each group, we report its density differences under the two measures. Note that, for $\textbf{x}$ under the graph affinity measure, its average degree is $\frac{W_D(S_{\textbf{x}})}{|S_{\textbf{x}}|}$. We also report the edge density difference, defined as $\frac{W_D(S)}{|S|^2}$ of each co-author group, since edge density can be regarded as a discrete version of graph affinity.

The results show that the research topics of the emerging groups are machine learning and security, which both are hot topics in recent years. As to the disappearing groups, Compiler \& Software System are all relatively mature areas of computer science, and, for the 3 Japanese Robotics research groups, it is known that recently Japanese researchers do not publish as many papers in international conferences as they did before. 

\begin{table*}
\centering
\begin{tabular}{|p{130mm}|c|}
\hline
List of Authors & Note \\ \hline
Feiping Nie(0.4428), Heng Huang(0.462), Chris H. Q. Ding(0.0230), Hua Wang(0.0717) & UTA Machine Learning \\ \hline
Lorrie Faith Cranor(0.1428), Nicolas Christin(0.1428), Blase Ur(0.1428), Richard Shay(0.1428), Saranga Komanduri(0.1428), Michelle L. Mazurek(0.1428), Lujo Bauer(0.1428) & CMU Privacy \& Security \\ \hline
Kensuke Harada, Kiyoshi Fujiwara, Fumio Kanehiro, Hirohisa Hirukawa, Shuuji Kajita, Kenji Kaneko & Japan Robotics 1 \\ \hline
Toshio Fukuda(0.5), Fumihito Arai(0.5) & Japan Robotics 2 \\ \hline
Fumio Kanehiro(0.1428), Shuuji Kajita(0.1428), Kenji Kaneko(0.1428), Kensuke Harada(0.1428), Kiyoshi Fujiwara(0.1428), Hirohisa Hirukawa(0.1428), Mitsuharu Morisawa(0.1428) & Japan Robotics 3 \\ \hline
Monica S. Lam, Katherine A. Yelick, Alok N. Choudhary, Michael L. Scott, James C. Browne, Marina C. Chen, Rudolf Eigenmann, Dennis Gannon, Charles Koelbel, Wei Li 0015, Thomas J. LeBlanc, David A. Padua, Constantine D. Polychronopoulos, Sanjay Ranka, Ian T. Foster, Carl Kesselman, Geoffrey Fox, Tomasz Haupt, Allen D. Malony, Janice E. Cuny, Joel H. Saltz, Alan Sussman & Compiler \& Software System \\ \hline
\end{tabular}
\caption{Co-author Groups}
\label{tab:group}
\end{table*}

\begin{table*}
\centering
\begin{tabular}{|c|c|c|c|c|c|c|c|c|c|}
\hline
Setting & $G_D$ Type & Density & Co-author Group & \#Authors & \tabincell{c}{Positive\\Clique?} & \tabincell{c}{Ave. Degree\\Difference} & \tabincell{c}{Approx.\\Ratio} & \tabincell{c}{Graph Affinity\\Difference} & \tabincell{c}{Edge Density\\Difference ($\frac{W_D(S)}{|S|^2}$)} \\ \hline
Weighted & Emerging & Average Degree & \tabincell{c}{UTA Machine\\Learning} & 4 & Yes & 81.5 & 2 & \NA & 20.375 \\ \hline
Weighted & Emerging & Graph Affinity & \tabincell{c}{UTA Machine\\Learning} & 4 & Yes & 81.5 & \NA & 23.167 & 20.375 \\ \hline
Weighted & Disappearing & Average Degree & \tabincell{c}{Japan\\Robotics 1} & 6 & Yes & 143 & 2 & \NA & 23.833 \\ \hline
Weighted & Disappearing & Graph Affinity & \tabincell{c}{Japan\\Robotics 2} & 2 & Yes & 50 & \NA & 50 & 50 \\ \hline
Discrete & Emerging & Average Degree & \tabincell{c}{CMU Privacy \\\& Security} & 7 & Yes & 12 & 2 & \NA & 1.714 \\ \hline
Discrete & Emerging & Graph Affinity & \tabincell{c}{CMU Privacy \\\& Security} & 7 & Yes & 12 & \NA & 1.714 & 1.714 \\ \hline
Discrete & Disappearing & Average Degree & \tabincell{c}{Compiler \&\\Software System} & 22 & Yes & 21.45 & 2 & \NA & 0.975 \\ \hline
Discrete & Disappearing & Graph Affinity & \tabincell{c}{Japan\\Robotics 3} & 8 & Yes & 14 & \NA & 1.714 & 1.714\\ \hline
\end{tabular}
\caption{Information of Co-Author Groups}
\label{tab:group_info}
\end{table*}
 
\subsection{\updates{Mining Emerging and Disappearing Data Mining Topics}} 
\updates{Using the same DBLP dataset, we extracted titles of papers published in some famous Data Mining venues including KDD, ICDM, SDM, PKDD, PAKDD, TKDE, TKDD and DMKD. Similar to~\cite{angel2012dense}, we built keyword association graphs from the paper titles. Unlike~\cite{angel2012dense}, we tried to identify emerging and disappearing data mining topics during 2008-2017, compared to the time period 1998-2007. Thus, we split all paper titles in two parts according to their publication years, and built two keyword association graphs $G_1$ (for 1998-2007) and $G_2$ (for 2008-2017). We removed all stop words and used the rest words in these paper titles as keywords. The edge weights of $G_1$ and $G_2$ were set based on the pairwise co-occurrences of keywords as suggested by~\cite{angel2012dense}. Specifically, for an edge between two keywords, we set its weight as 100 times the percentage of paper titles containing both the keywords. Statistics of the difference graphs can be found in Table~\ref{tab:gd_sta} (the DM dataset).}

\updates{This time again all DCSGA algorithms found the same emerging topic \textbf{\{social (0.5), networks (0.5)\}} and the same disappearing topic \textbf{\{mining (0.12), association (0.44), rules (0.44)\}}. Our DCSGreedy algorithm for solving DCSAD also found the disappearing topic \{mining, association, rules\}. We skip the emerging topic w.r.t. the average degree measure, because DCSGreedy found a large set of 38 keywords which lacks interpretability. Since a research topic/story often only has a few keywords, the graph affinity which prefers small and densely connected subgraphs is a more proper density measure in this task compared to the average degree. In~\cite{angel2012dense}, Angel~\textit{et~al.} also suggested to use small and dense subgraphs for identifying stories in text data.
}
 
\updates{To further demonstrate the effectiveness of applying DCS in identifying emerging/disappearing research topics, we also display the top results returned by our SEACD+Refinement algorithm. Remember this algorithm does initializations using every vertex in $G_D$ and returns multiple positive cliques in $G_D$. We removed the duplicate cliques and the cliques that are sub-graphs of other cliques found. We list the top-5 positive cliques with the highest graph affinity difference found by the SEACD+Refinement in Table~\ref{tab:DM_EmDis}. }

\updates{From the results we can find that our DCSGA algorithms are very effective. Social networks, matrix factorization, semi-supervised learning and unsupervised feature selection all became hot topics only in recent years, and they were not that popular in early years. Moreover, due to the need from industry and the development of computation power, large scale is turning into one of the most important concerns in data mining research. For the disappearing topics, association rule mining, support vector machines, inductive logic programming and intrusion detection are all relatively mature research topics which were majorly investigated in early years. ``Knowledge discovery'' used to be a popularly adopted term when data mining as a research area arouse.}

\updates{What's more, we also report the top-5 topics in $G_1$ and $G_2$ in Table~\ref{tab:DM_G1G2}. Since average degree density measure prefers large subgraphs and is not very proper for identifying topics/stories, we do not report the top topics w.r.t.\ average degree. The aim of displaying such results is to show the necessity of applying DCS to find emerging/disappearing topics. If we mine emerging/disappearing topics only in one graph like~\cite{angel2012dense} does, the results may be not effective. For example, if we only consider $G_2$ to mine emerging topics, we would find \{time (0.5), series (0.5)\} and \{feature (0.5), selection (0.5)\}. However, \{time (0.5), series (0.5)\} and \{feature (0.5), selection (0.5)\} were hot topics before 2008 so they were not emerging topics during 2008-2017. The topic \{time (0.5), series (0.5)\} even cooled down in the last ten years, since its graph affinity density dropped from 1.185 (in $G_1$) to 1.049 (in $G_2$) according to our calculation.}
\begin{table}
\centering
\begin{tabular}{|c|c|c|}
\hline
\multirow{2}*{Rank} & \multicolumn{2}{|c|}{Keyword Set/Topic}  \\ \cline{2-3}
 & Emerging & Disappearing\\ \hline
1 & \{social (0.5), networks (0.5)\} & \tabincell{c}{\{mining (0.12), association (0.45),\\ rules (0.43)\}} \\ \hline
2 & \{large (0.5), scale (0.5)\} & \{knowledge (0.5), discovery (0.5)	\} \\ \hline
3 & \{matrix (0.5), factorization (0.5)\} & \tabincell{c}{\{support (0.39), vector (0.38),\\ machines (0.23)	\}} \\ \hline
4 & \tabincell{c}{\{semi (0.45), supervised (0.45),\\ learning (0.1)\}} & \tabincell{c}{\{logic (0.36), inductive (0.26),\\ programming (0.38)\}} \\ \hline
5 & \tabincell{c}{\{unsupervised (0.34), feature (0.29),\\ selection (0.27)\}} & \{intrusion (0.5), detection (0.5)\} \\ \hline
\end{tabular}
\caption{\updates{Top 5 Emerging/Disappearing Topics w.r.t.\ Graph Affinity}}
\label{tab:DM_EmDis}
\end{table}

\begin{table}
\centering
\begin{tabular}{|c|c|c|}
\hline
\multirow{2}*{Rank} & \multicolumn{2}{|c|}{Keyword Set/Topic}  \\ \cline{2-3}
 & $G_1$ (1998-2007) & $G_2$ (2008-2017)\\ \hline
1 & \{time (0.5), series (0.5)\} & \{social (0.5), networks (0.5)\} \\ \hline
2 & \tabincell{c}{\{support (0.41), vector (0.41),\\ machines (0.18)\}} & \{time (0.5), series (0.5)\} \\ \hline
3 & \{feature (0.5), selection (0.5)\} & \{large (0.5), scale (0.5)\} \\ \hline
4 & \{decision (0.5), trees (0.5)\} & \{feature (0.5), selection (0.5)\} \\ \hline
5 & \{nearest (0.5), neighbor (0.5)\} & \tabincell{c}{semi (0.46), supervised (0.47),\\ learning (0.07)\}} \\ \hline
\end{tabular}
\caption{\updates{Top 5 Topics w.r.t.\ Graph Affinity}}
\label{tab:DM_G1G2}
\end{table}

\nop{
\begin{table}
\centering
\begin{tabular}{|c|c|c|c|c|}
\hline
\multirow{2}*{Rank} & \multicolumn{2}{|c|}{$G_1$ (1998-2007)} & \multicolumn{2}{|c|}{$G_2$ (2008-2017)}  \\ \cline{2-5}
 & Keywords Set/Topic & \tabincell{c}{Graph\\Affinity} & Keywords Set/Topic & \tabincell{c}{Graph\\Affinity} \\ \hline
 1 & \{time (0.5), series (0.5)\} & 1.00 & \{social (0.5), networks (0.5)\} & 1.00 \\ \hline
2 & \tabincell{c}{\{support (0.41), vector (0.41),\\ machines (0.18)\}} & 1.00 & \{time (0.5), series (0.5)\} & 1.00 \\ \hline
3 & \{feature (0.5), selection (0.5)\} & 1.00 & \{large (0.5), scale (0.5)\} & 1.00  \\ \hline
4 & \{decision (0.5), trees (0.5)\} & 1.00 & \{feature (0.5), selection (0.5)\} & 1.00  \\ \hline
5 & \{nearest (0.5), neighbor (0.5)\} & 1.00 & \tabincell{c}{\{semi (0.46), supervised (0.47),\\ learning (0.07)\}} & 1.00  \\ \hline
\end{tabular}
\caption{\updates{Top 5 Topics w.r.t.\ Graph Affinity}}
\label{tab:DM_G1G2}
\end{table}
}
\subsection{Efficiency Comparison}

Limited by space, we focus on the running time of the $\textbf{DCSGA}$ algorithms, since all $\textbf{DCSAD}$ algorithms have quasi-linear time complexity $O((m_1+m_2+n)\log{n})$, and are efficient and scalable in practice. 

Besides the above DCS mining tasks, to compare the efficiency of the algorithms, we also employed several other data sets whose statistics can be found in Table~\ref{tab:gd_sta}. How these datasets were generated and the description of experiments on these datasets please refer to the Appendix.
\nop{
Besides the above DCS mining tasks, to compare the efficiency of the algorithms, we also employed two large data sets, DBLP-C and Actor. The DBLP-C data set contains timestamped co-authorship records. We split all co-authorship records into two almost even parts by setting a specific timestamp as the separation timestamp. Then the two parts were used to build two co-author graphs $G_1$ and $G_2$, where the weight of an edge between two vertices (authors) is the number of collaborations. Similar to the Emerging/Disappearing co-author group mining task, we adopted the Weighted and Discrete settings to build the difference graph $G_D$. The actor data set is a collaboration network of actors, where the weight of an edge between two vertices (actors) is the number of collaborations. We directly used this actor collaboration network as a difference graph, since as pointed out in Section~\ref{sec:refine}, our \textbf{DCSGA} algorithms are also competitive solutions to traditional graph affinity maximization on weighted graphs. For the Actor difference graph, we also tried the Weighted setting and the Discrete setting, where in the Discrete setting we set edge weights $D(u,v)=10$ if $D(u,v)$ originally was greater than 10. The statistics of DBLP-C and Actor difference graphs can be found in Table~\ref{tab:gd_sta}. Table~\ref{tab:large_ga} reports the DCS found, where all $\textbf{DCSGA}$ algorithms again found the same DCS every time.
}
\begin{table*}[t]
\centering
\begin{tabular}{|c|c|c|c|c|c|c|}
\hline
Data & Setting & $G_D$ Type & NewSEA & \tabincell{c}{SEACD+\\Refine} & \tabincell{c}{SEA+\\Refine} & \tabincell{c}{\#Errors\\in SEA} \\ \hline
DBLP & Weighted & Emerging & 0.05 & 3.2 & 14.3 & 1 \\ \hline
DBLP & Weighted & Disappearing & 0.05 & 3.2 & 13.7 & 1 \\ \hline
DBLP & Discrete & Emerging & 0.06 & 2.9 & 7.3 & 2 \\ \hline
DBLP & Discrete & Disappearing & 0.06 & 2.9 & 6.8 & 0 \\ \hline
DM & \NA & Emerging & 0.35 & 14.1 & 185.3 & 0 \\ \hline
DM & \NA & Disappearing & 0.21 & 6.9 & 36.3 & 0 \\ \hline
Wiki & \NA & Consistent & 56.6 & 452 & 36121 & 80 \\ \hline
Wiki & \NA & Conflicting & 23.8 & 110 & 7703 & 211 \\ \hline
Movie & \NA & \tabincell{c}{Interest$-$\\Social} & 16.3 & 29.6 & 580.6 & 1 \\ \hline
Movie & \NA & \tabincell{c}{Social$-$\\Interest} & 23.1 & 32.7 & 404.8 & 1 \\ \hline
Book & \NA & \tabincell{c}{Interest$-$\\Social} & 2.02 & 14.5 & 53.2 & 0 \\ \hline
Book & \NA & \tabincell{c}{Social$-$\\Interest} & 20.9 & 32.7 & 397 & 0 \\ \hline
DBLP-C & Weighted & \NA & 2.01 & 8054 & 23090 & 118 \\ \hline
DBLP-C & Discrete & \NA & 12.3 & 7678 & 22837 & 131 \\ \hline
Actor & Weighted & \NA & 2.3 & 2249 & 73671 & 321 \\ \hline
Actor & Discrete & \NA & 155 & 2574 & 124132 & 4419 \\ \hline
\end{tabular}
\caption{Running time in seconds.}
\label{tab:time}
\end{table*}

Table~\ref{tab:time} reports the running time of each $\textbf{DCSGA}$ algorithm on each data set. Since we set different convergence conditions for the Shrink stage of each algorithm, one may wonder whether the convergence condition for SEA+Refine is too strict and makes SEA+Refine not as efficient as the other two algorithms. Thus, we also report the number of errors made by SEA+Refine in the Expansion stages. Note that the errors in Expansion are caused by that the Shrink stage cannot reach a local KKT point. From Table~\ref{tab:time} we find that the SEA+Refine algorithm often made mistakes in the Expansion stage, which means the convergence condition for the Shrink stage of SEA+Refine is still too loose to achieve a local KKT point. It is worth noting that the two algorithms using our coordinate descent algorithm in the Shrink stage, NewSEA and SEACD+Refine, never made mistakes in the Expansion stage. We also find that our NewSEA algorithm often is much faster than the other two algorithms. Note that the only difference between NewSEA and SEACD+Refine is the smart initialization heuristic. Compared to SEACD+Refine, the smart initialization heuristic sometimes brings us a speed up of 3 orders of magnitude. Moreover, SEACD+Refine is always faster than SEA+Refine, sometimes 80 times faster. It seems when the input $G_{D^+}$ is sparse, SEACD+Refine and SEA+Refine are close in efficiency. When $G_{D^+}$ becomes denser, the gap in efficiency gets larger. The Expansion error rate (defined by $\frac{\text{\#Errors in SEA}}{n}$) seems correlated with how dense $G_{D^+}$ is. The results are shown in Fig.~\ref{fig:avrng}, where $m^+/n$ measures how dense $G_{D^+}$ is, and $m^+$ is the number of edges in $G_{D^+}$.

\begin{figure}[t]
  \centering
    \subfigure[Speed-Up]{
      \includegraphics[width=.17\textwidth]{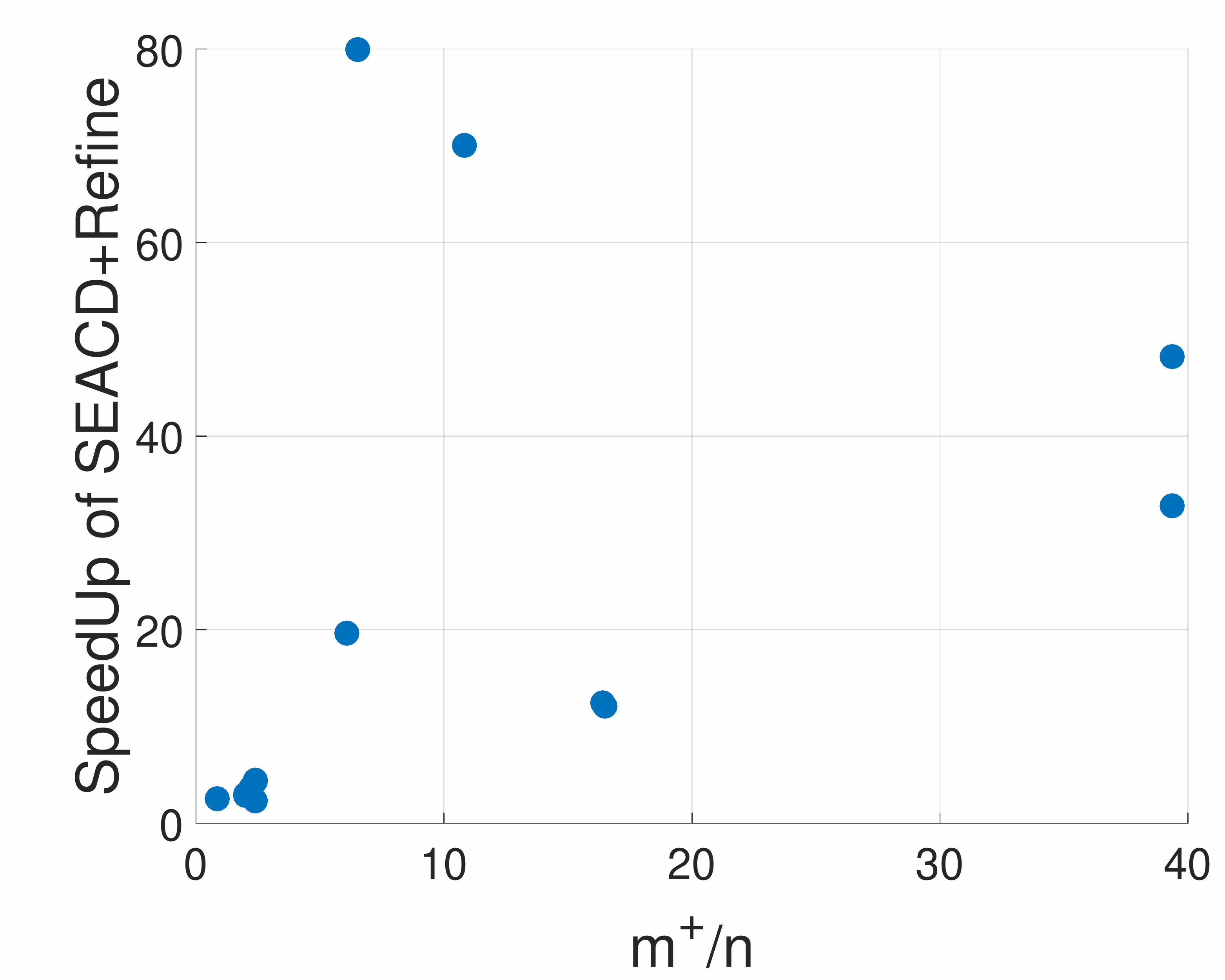}
    }
    \subfigure[Errors of SEA]{
      \includegraphics[width=.17\textwidth]{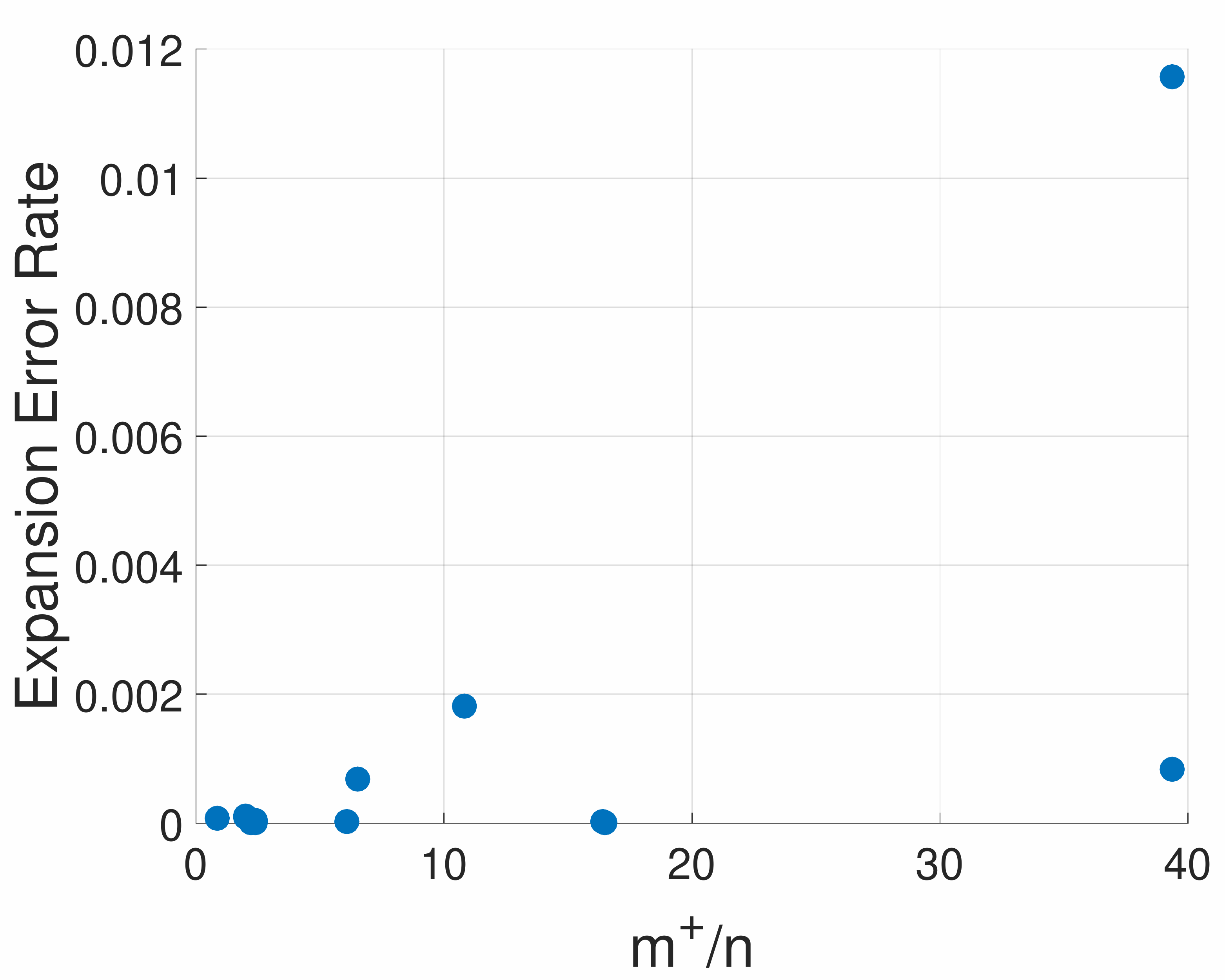}
    }
    \caption{SpeedUp of SEACD+Refine and Errors in Expansions of SEA+Refine}
    \label{fig:avrng}
\end{figure}

\subsection{\updates{Comparison with EgoScan~\cite{cadena2016dense}}}
\updates{Both DCSAD and DCSGA are new problems that were not discussed in literature before, and this paper focuses on algorithmic solutions to the two problems, so there are no very suitable baselines for our algorithms. However, in this section, we still compare our DCS mining algorithms with the EgoScan algorithm in~\cite{cadena2016dense}, which is the work closest to ours in literature. The objective of EgoScan is to maximize $W_D(S)$ subject to $S \subseteq V$ on the difference graph $G_D$. }

\updates{We ran the EgoScan algorithm\footnote{We thank the authors of~\cite{cadena2016dense} for providing us the code of EgoScan.} on the datasets used in our experiments. Unfortunately, since EgoScan needs a Semi-Definite Programming (SDP) solver as a frequently used subroutine, and the SDP solver is really slow and consumes too much memory when ego nets of vertices are large (having more than thousands of vertices), we only got results on the 4 DBLP co-author difference graphs that we used to draw emerging/disappearing co-author groups. For the 4 graphs, EgoScan always spent more than 100 seconds to finish. For other datasets, either EgoScan could not finish running in one day or the memory (16GB) of our machine was not enough for the SDP solver. The high computational cost is actually one drawback of applying EgoScan in practice.}

\updates{We display the results of running EgoScan on the DBLP co-author data. Since all co-author groups found by EgoScan have at least 44 authors, we cannot list all the authors. We only show statistics of these co-author groups. From Table~\ref{tab:egoscan_sta} and referring to Table~\ref{tab:group_info} which shows statistics of the author groups found by our DCS algorithms, we find that our DCS algorithms are much better than EgoScan in finding DCS w.r.t.\ average degree and edge density. Moreover, subgraphs found by EgoScan are all big, even bigger than the subgraphs found by our DCSGreedy algorithms.}

\begin{table*}[t]
\centering
\begin{tabular}{|c|c|c|c|c|c|c|c|}
\hline
Setting & $G_D$ Type & \#Authors & \#Edges & Positive Clique? & Ave. Degree Difference &  Edge Density Difference ($\frac{W_D(S)}{|S|^2}$) \\ \hline
Weighted & Emerging & 82 & 473 & No & 26.95 & 0.3287 \\ \hline
Weighted & Disappearing & 59 & 311 & No & 45.39 & 0.7693 \\ \hline
Discrete & Emerging & 44 & 124 & No & 7.46 & 0.1694 \\ \hline
Discrete & Disappearing & 80 & 527 & No & 13.8 & 0.1725 \\ \hline
\end{tabular}
\caption{\updates{Statistics of Co-Author Groups (Subgraphs) Found by EgoScan}}
\label{tab:egoscan_sta}
\end{table*}

\updates{We also compare our DCS algorithms with EgoScan in finding subgraphs w.r.t.\ the total edge weight difference $W_D(S)$, which is shown in Table~\ref{tab:wd_comp}. Note that the total edge weigh difference of a solution $\textbf{x}$ returned by our NewSEA algorithm is defined as $W_D(S_{\textbf{x}})$. Under the evaluation metric of total edge weight difference, EgoScan performs much better than our DCS algorithms.}
\begin{table}[t]
\centering
\begin{tabular}{|c|c|c|c|c|}
\hline
Setting & $G_D$ Type & DCSGreedy & NewSEA ($W_D(S_{\textbf{x}})$) & EgoScan \\ \hline
Weighted & Emerging & 326 & 326 & 2210 \\ \hline
Weighted & Disappearing & 858 & 100 & 2678 \\ \hline
Discrete & Emerging & 84 & 84 & 328 \\ \hline
Discrete & Disappearing & 472 & 112 & 1104 \\ \hline
\end{tabular}
\caption{\updates{Total edge weight difference ($W_D(S)$) of Co-Author Groups Found by DCS algorithms and EgoScan}}
\label{tab:wd_comp}
\end{table}

\updates{Table~\ref{tab:egoscan_sta},~\ref{tab:wd_comp} and~\ref{tab:group_info} show that DCS w.r.t.\ different measures could be very different. We have the following rough suggestions for deciding which measures to use in practice: (1) if users prefer small DCS and good interpretability, we should take graph affinity as the density measure and apply our NewSEA algorithm, since it always returns a positive clique where for every pair of vertices, their connection in $G_2$ is tighter than their connection in $G_1$; (2) if users prefer a medium sized subgraph, then average degree should be the measure and we apply our DCSGreedy algorithm; (3) If users want a even larger subgraph, total edge weight maybe the suitable measure because it seems that such a measure encourages even bigger subgraphs than average degree.}

\section{Conclusion}\label{sec:con}
In this paper, we studied the Density Contrast Subgraph problem that have interesting applications in practice. Two popularly adopted graph density measures, average degree and graph affinity, were considered. We proved the hardness of the DCS problem under the two measures, and devised algorithms that work well in practice for finding DCS under both density measures. We reported a series of experiments on both real and synthetic datasets and demonstrated the effectiveness and efficiency of our algorithms.

There are some interesting future directions. For example, our methods are based on graph density, but density sometimes cannot reflect how ``dissimilar'' a subgraph looks in two graphs. Thus, how to extract subgraphs that are dissimilar in two graphs with respect to some graph similarity measures~\cite{koutra2013deltacon} is interesting. Also, our methods only mine one DCS with the greatest density difference, how to mine multiple subgraphs with big density difference is another interesting direction.

\bibliographystyle{abbrv}
\bibliography{DCS}

\section*{Appendix}\label{sec:appendix}

\subsection{The SEA Algorithm}

The SEA algorithm~\cite{liu2013fast} solves Eq.~\ref{eq:DCS_ga} when the symmetric matrix $D$ only has non-negative entries. The strategy of SEA is to iteratively find a local KKT point (Shrink stage) and expand it to more vertices (Expansion stage) until convergence. 

\noindent \textbf{Shrink Stage.} To find a local KKT point on a set of vertices $S$, a replicator dynamic is exploited. The replicator equation is
\begin{equation}\label{eq:sea_shrink}
    x_i(t+1)=x_i(t)\frac{(D\textbf{x})_i}{\textbf{x}(t)^\top D \textbf{x}(t)}, ~i \in S
\end{equation}
where $x_i(t)$ is the value of $x_i$ in the $t$-th iteration. To make this replicator dynamic converge, $D$ should be non-negative.

\smallskip

\noindent \textbf{Expansion Stage.} In the Expansion stage, SEA firstly find the set $Z$ as Algorithm~\ref{alg:seacd} does in Line~\ref{line:expansion}. According to Eq.~\ref{eq:sea_shrink}, if $x_i$ at the beginning of the replicator dynamic is 0, it will stay 0 forever. Thus, SEA needs to give a positive initial value $x_v$ to each vertex $v \in Z$. To do that, we first define the \textbf{$\gamma$} vector,
\begin{equation*}
    \gamma_i = 
        \begin{cases}
            0~~~~~~~~~~~~~~~~~~~~~i \notin Z \\
            \nabla_if_D(\textbf{x})-f_D(\textbf{x})~~~i \in Z
        \end{cases}
\end{equation*}
where $\textbf{x}$ is the local KKT point to be expanded by $Z$. Let $s=\sum_{i \in Z}{\gamma_i}$, $\zeta=\sum_{i \in Z}{\gamma_i^2}$ and $\omega=\sum_{i,j \in Z}{\gamma_i\gamma_jD(i,j)}$. The Expansion stage updates $\textbf{x}$ along the direction $\textbf{b}$, where
\begin{equation*}
     b_i = 
        \begin{cases}
            -x_is~~~~~~~~~i \in S_{\textbf{x}} \\
            \gamma_i~~~~~~~~~~~~i \in Z
        \end{cases}
\end{equation*}
Let $f_D(\textbf{x})=\bar{\lambda}$. We maximize $f_D(\textbf{x}+\tau \textbf{b})-f_D(\textbf{x})=-(\bar{\lambda}s^2+2s\zeta-\omega)\tau^2-2\zeta\tau$ over $\tau$. The best $\tau$ can be found analytically. Let $a=\bar{\lambda}s^2+2s\zeta-\omega$, when $a \leq 0$ we set $\tau=\frac{1}{s}$, and we set $\tau=\min\{\frac{1}{s},-\frac{1}{a}\}$ otherwise. Then $\textbf{x}$ is updated to $\textbf{x}+\tau \textbf{b}$.

\subsection{More Experimental Results}\label{sec:more_exp}

\begin{table*}[h]
\centering
\begin{tabular}{|c|c|c|c|c|c|c|c|c|c|c|c|}
\hline
\multirow{2}*{$G_D$ Type} & \multicolumn{4}{|c|}{DCSGreedy} & \multicolumn{3}{|c|}{$G_D$ Only} & \multicolumn{3}{|c|}{$G_{D^+}$ Only}  \\ \cline{2-11}
& \#Users & \tabincell{c}{Ave. Degree\\Difference} & \tabincell{c}{Approx.\\Ratio} & \tabincell{c}{Positive\\Clique?} & \#Users & \tabincell{c}{Average\\Degree} & \tabincell{c}{Positive\\Clique?} & \#Users & \tabincell{c}{Average\\Degree} & \tabincell{c}{Positive\\Clique?} \\ \hline 
Consistent & 937 & 398.71 & 2.13 & No & 1013 & 345.25 & No & 937 & 398.71 & No \\ \hline
Conflicting & 222 & 335.03 & 2.06 & No & 222 & 335.03 & No & 230 & 332.24 & No \\ \hline
\end{tabular}
\caption{Effectiveness Comparison of Mined DCSs with respect to Average Degree on Wiki Data}
\label{tab:wiki_ad}
\end{table*}

\begin{table}[H]
\centering
\begin{tabular}{|c|c|c|c|c|c|c|c|}
\hline
$G_D$ Type & \#Users & \tabincell{c}{Graph Affinity\\Difference} & \tabincell{c}{Edge Density\\Difference$\frac{W_D(S_{\textbf{x}})}{|S_{\textbf{x}}|^2}$} \\ \hline
Consistent & 5 & 6.901 & 6.845 \\ \hline
Conflicting & 6 & 6.456 & 6.209 \\ \hline
\end{tabular}
\caption{DCS with respect to Graph Affinity on Wiki Data}
\label{tab:wiki_ga}
\end{table}

\begin{table*}[h]
\centering
\begin{tabular}{|c|c|c|c|c|c|c|c|c|c|c|c|c|}
\hline
\multirow{2}*{Interest} & \multirow{2}*{$G_D$ Type} & \multicolumn{4}{|c|}{DCSGreedy} & \multicolumn{3}{|c|}{$G_D$ Only} & \multicolumn{3}{|c|}{$G_{D^+}$ Only}  \\ \cline{3-12}
& & \#Users & \tabincell{c}{Ave. Degree\\Difference} & \tabincell{c}{Approx.\\Ratio} & \tabincell{c}{Positive\\Clique?} & \#Users & \tabincell{c}{Average\\Degree} & \tabincell{c}{Positive\\Clique?} & \#Users & \tabincell{c}{Average\\Degree} & \tabincell{c}{Positive\\Clique?} \\ \hline 
Movie & Interest$-$Social & 968 & 176.002 & 2.05 & No & 968 & 176.002 & No & 1003 & 175.958 & No \\ \hline
Movie & Social$-$Interest & 4047 & 68.288 & 2.11 & No & 4047 & 68.288 & No & 4473 & 62.351 & No \\ \hline
Book & Interest$-$Social & 610 & 43.190 & 2.12 & No & 610 & 43.190 & No & 719 & 42.754 & No \\ \hline
Book & Social$-$Interest & 4175 & 71.280 & 2.03 & No & 4175 & 71.280 & No & 4403 & 70.825 & No\\ \hline
\end{tabular}
\caption{Effectiveness Comparison of Mined DCSs with respect to Average Degree on Douban Data} 
\label{tab:douban_ad}
\end{table*}

\begin{table}[H]
\centering
\begin{tabular}{|c|c|c|c|c|c|c|c|}
\hline
Interest & $G_D$ Type & \#Users & \tabincell{c}{Graph Affinity\\Difference} & \tabincell{c}{Edge Density\\Difference$\frac{W_D(S_{\textbf{x}})}{|S_{\textbf{x}}|^2}$} \\ \hline
Movie & Interest$-$Social & 32 & 0.969 & 0.969 \\ \hline
Movie & Social$-$Interest & 18 & 0.944 & 0.944 \\ \hline
Book  & Interest$-$Social & 14 & 0.929 & 0.929 \\ \hline
Book  & Social$-$Interest & 22 & 0.955 & 0.955 \\ \hline
\end{tabular}
\caption{DCS with respect to Graph Affinity on Douban Data}
\label{tab:douban_ga}
\end{table}

\begin{table}[H]
\centering
\begin{tabular}{|c|c|c|c|c|c|c|c|}
\hline
Data & Setting & \#Users & \tabincell{c}{Graph Affinity\\Difference} & \tabincell{c}{Edge Density\\Difference$\frac{W_D(S_{\textbf{x}})}{|S_{\textbf{x}}|^2}$} \\ \hline
DBLP-C & Weighted & 2 & 200 & 200 \\ \hline
DBLP-C & Discrete & 26 & 1.919 & 1.917 \\ \hline
Actor  & Weighted & 3 & 108.25 & 98.44 \\ \hline
Actor  & Discrete & 21 & 6.46 & 6.21 \\ \hline
\end{tabular}
\caption{DCS with respect to Graph Affinity on DBLP-C and Actor Data}
\label{tab:large_ga}
\end{table}

\subsubsection{Extracting Consistent and Conflicting Editor Groups}
We also tested our DCS mining algorithm on a Wikipedia network with editor (user) interaction data (\url{http://konect.uni-koblenz.de/networks/wikiconflict}). This wiki dataset has two weighted networks, a positive interaction network $G_1$ and a negative interaction network $G_2$, where vertices are editors of Wikipedia pages. An example of a negative interaction is when one user revert the edit of another user.  

Similar to the emerging/disappearing co-author group mining task, we also mined two kinds of editor groups. The first kind is the editor group whose consistency in editing is much more than their conflict and the second kind is the opposite. We call the first kind of editor group the Consistent group and the second kind the Conflicting group. For mining the Consistency group and the Conflicting group, we ran the algorithms on the Consistent $G_D$, which is set to $G_1-G_2$, and the Conflicting $G_D$, which equals $G_2-G_1$, respectively.

Tables~\ref{tab:wiki_ad} and~\ref{tab:wiki_ga} show the results. The 3 \textbf{DCSGA} algorithms again produced the same result. Unlike in the DBLP dataset, where $\textbf{DCSAD}$ algorithms all find positive cliques as DCS, this time none of the subgraphs produced by them is a positive clique. Moreover, all DCS with respect to average degree are significantly larger in size than the DCS with respect to graph affinity. This is similar to the fact that for mining dense subgraphs in a single graph, average degree density encourages large subgraphs while graph affinity density prefers subgraphs with small size~\cite{wang2016tradeoffs, mitzenmacher2015scalable}.

\subsubsection{Mining DCS in Douban Network}
We also applied our DCS mining algorithms on a Douban dataset~\cite{xu2012towards}, where Douban is a famous content-sharing online social network in China. The Douban dataset contains a user social network, and movie ratings and book ratings of every user. We extract users who have rated at least 50 movies and 20 books, and the induced subgraph of these users in the social network is recorded as $G_1$. To build the interest similarity graph $G_2$, we utilize user ratings data. Specifically, we computed the Jaccard similarity of movies/books rated by two users $u$ and $v$ if $u$ and $v$ are within 2 hops in $G_1$. For the movie similarity graph, we built an edge $(u,v)$ if the Jaccard similarity between the movie lists rated by them is greater than 0.2. For the book similarity graph, this threshold value is set to 0.1, because Book ratings are sparser than Movie ratings. Both $G_1$ and $G_2$ are uniformly weighted graph, that is, edge weights are all 1.

Still, we built two types of $G_D$, Interest$-$Social ($G_2-G_1$) and Social$-$Interest ($G_1-G_2$). The Interest$-$Social $G_D$ is for mining DCS whose density in the interest graph minus its density in the social graph is maximized, while the Social$-$Interest $G_D$ is for mining the opposite kind of DCS. Before we show the results of mining DCS, let us first look at the statistics of the difference graphs of the Douban data in Table~\ref{tab:gd_sta}. Even we set a lower threshold on the Jaccard similarity for building the Book interest similarity graph, the Interest$-$Social difference graph of Book still has substantially less positive weighted edges than the Interest$-$Social difference graph of Movie. Moreover, no matter the interest is Movie or Book, the Interest$-$Social difference graph has much less positive edges than the Social$-$Interest difference graph. Our experimental results will show more interesting findings that are not reflected just by the statistics of difference graphs.

Table~\ref{tab:douban_ad} and Table~\ref{tab:douban_ga} show the results. Similar to the wiki dataset, the \textbf{DCSAD} algorithms all find big subgraphs. All \textbf{DCSGA} algorithms again extract the same embedding (subgraph) in every type of the difference graph.

One interesting finding from Table~\ref{tab:douban_ad} and Table~\ref{tab:douban_ga} is that no matter what graph density is used, for the Movie interest, the DCS from the Interest$-$Social $G_D$ is denser (has a greater density difference) than the DCS from the Social$-$Interest $G_D$, which has more positive edges than the Interest$-$Social $G_D$. However, for the Book interest, we get the opposite result. 

Remember that the SEACD+Refinement algorithm do initializations using every vertex in $G_D$, and it actually finds multiple positive cliques in $G_D$. Thus, we also report the statistics of cliques found by the SEACD+Refinement algorithm. We removed the duplicate cliques and the cliques that are subgraphs of other cliques found. Fig~\ref{fig:clique} shows the results, where for the Movie interest, we report the counts of $k$-cliques found where $k \geq 10$ and for the Book interest we report the counts of $k$-cliques found where $k \geq 8$. Although the Social$-$Interest $G_D$ has much more positive edges than the Interest$-$Social $G_D$ for the Movie interest, the Social$-$Interest $G_D$ has more and larger positive cliques than the Interest$-$Social $G_D$. Again, for the Book data, the situation is the opposite. 

\begin{figure}[h]
  \centering
    \subfigure[Movie]{
      \includegraphics[width=.2\textwidth]{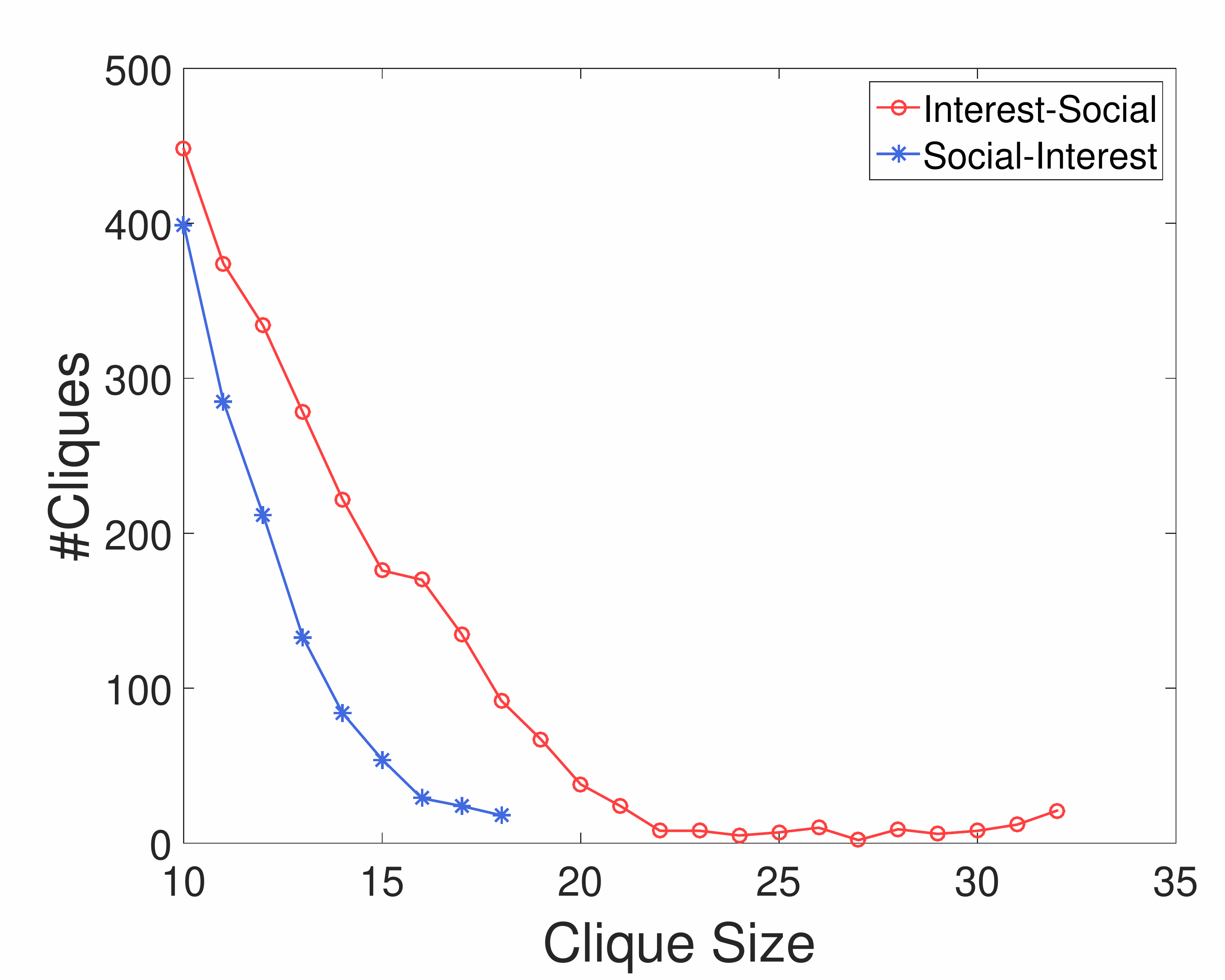}
    }
    \subfigure[Book]{
      \includegraphics[width=.2\textwidth]{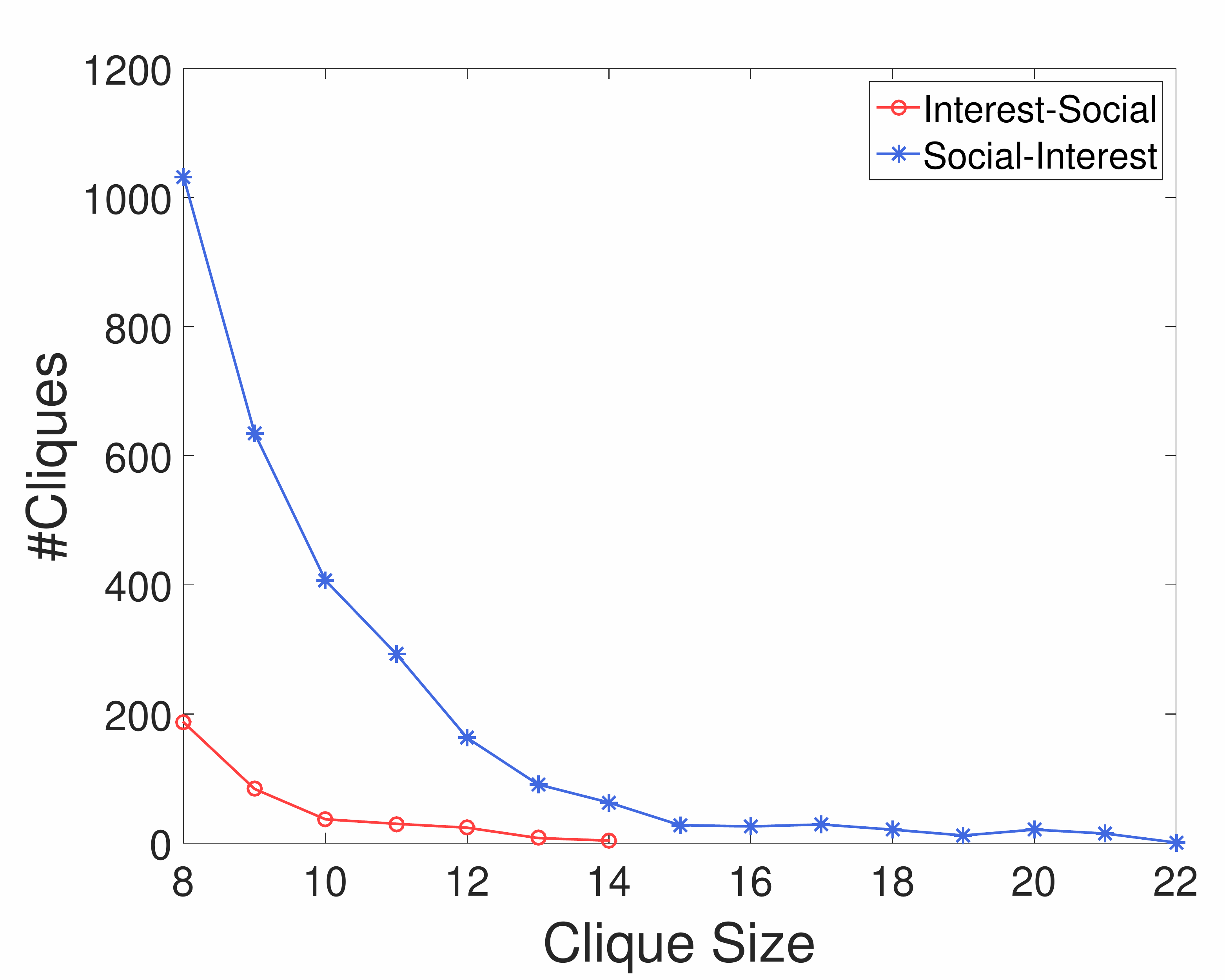}
    }
    \caption{Clique Counts of Douban Data}
    \label{fig:clique}
\end{figure} 

The results of mining DCS from the Douban data suggest that the formation of the Douban social network may depend more on users' interest similarity on Movie than on Book.

\subsubsection{DBLP-C and Actor Datasets}
To compare the efficiency of the algorithms, we also employed two large data sets, DBLP-C and Actor. The DBLP-C data set contains timestamped co-authorship records. We split all co-authorship records into two almost even parts by setting a specific timestamp as the separation timestamp. Then the two parts were used to build two co-author graphs $G_1$ and $G_2$, where the weight of an edge between two vertices (authors) is the number of collaborations. Similar to the Emerging/Disappearing co-author group mining task, we adopted the Weighted and Discrete settings to build the difference graph $G_D$. The actor data set is a collaboration network of actors, where the weight of an edge between two vertices (actors) is the number of collaborations. We directly used this actor collaboration network as a difference graph, since as pointed out in Section~\ref{sec:refine}, our \textbf{DCSGA} algorithms are also competitive solutions to traditional graph affinity maximization on weighted graphs. For the Actor difference graph, we also tried the Weighted setting and the Discrete setting, where in the Discrete setting we set edge weights $D(u,v)=10$ if $D(u,v)$ originally was greater than 10. The statistics of DBLP-C and Actor difference graphs can be found in Table~\ref{tab:gd_sta}. Table~\ref{tab:large_ga} reports the DCS found, where all $\textbf{DCSGA}$ algorithms again found the same DCS every time.

\end{document}